\documentclass[11pt,twoside,letterpaper]{article}
\usepackage{hyperref}
\usepackage{amsmath}
\usepackage{amssymb}
\usepackage{amsfonts}
\usepackage{amsthm}
\usepackage{comment}
\usepackage{color}
\usepackage{verbatim}
\usepackage{subfig}
\usepackage{graphicx}
\usepackage{float}
\usepackage{wrapfig}

\def\cT{{\mathcal {T}}}
\def\p{\partial}

\def\R{{ \mathbb{R}}}

\def\N{{\mathbb{N}}}
\def\Z{{\mathbb{Z}}}

\def\cR{{\mathcal R}}

\def\di{\diamond}
\def\di{\diamond}

\newtheorem{theorem}{Theorem}[section]
\newtheorem{lemma}[theorem]{Lemma}
\newtheorem{proposition}[theorem]{Proposition}

\newtheorem{corollary}[theorem]{Corollary}

\theoremstyle{definition}
\newtheorem{definition}[theorem]{Definition}

\newtheorem{remark}[theorem]{Remark}

\newtheorem*{remark*}{Remark}

\numberwithin{figure}{section} \numberwithin{equation}{section}

\title{A Kermack--McKendrick type epidemic model with
double threshold phenomenon (and a possible application
to Covid-19)}

\author{
Joan Ponce\thanks{({\tt jponce15@asu.edu})} and Horst R. Thieme\thanks{({\tt hthieme@asu.edu})}
\\
School of Mathematical and Statistical Sciences
\\
Arizona State University,
Tempe, AZ 85287-1804, USA
}
\date{Sep 14, 2024}

\begin{document}

\maketitle


\begin{abstract}
The  suggestion by K.L. Cooke (1967) that infected individuals become
infective if they are exposed often enough
 for a natural disease resistance to be  overcome
 is built into
a Kermack-McKendrick type epidemic model
with infectivity age.
Both the case that the resistance may be
the same for all hosts and the case that it is distributed
among the host population are considered.
In addition to the familiar threshold behavior of
the final size of the epidemic with respect
to a basic reproductive number,
an Allee effect is generated with respect to the final
cumulative force of infection
exerted by the initial infectives.
This offers a deterministic  explanation  why  geographic areas
that appear to be epidemiologically similar
have epidemic outbreaks of quite different severity.
\end{abstract}

{\bf Keywords:} force of infection, state-dependent
delay, host heterogeneity, density-dependent
incidence, minimal solutions


\section{Introduction}
\label{sec:intro}

The emergence of Covid-19 has rekindled the interest
in Kermack--McKen\-drick type epidemic models.
See \cite{BoDI1, BoDI2, DeGrMa, DeMa, DiOt, DuMa, dIM, dMI, Eik, Gom, Ina23, KrSc, LiMa, LiMaWe, LuSt-free, LuSt-bol, SaVH, Tka1, Tka2} and the references therein.
It also raised questions as to whether models of
this type need to be modified in order to
explain some of the observations.
For instance, in somewhat comparable cities  in Switzerland, outbreaks
of very different size have been observed.
 ``Whereas in Geneva and Ticino 10
per cent and more of
the population was infected,  the level of infections was much lower in Zurich -- probably by a factor of 10." \cite{LuSt-free}

Stochastic modelers of epidemics will not be surprised by this phenomenon. If $\cR_0$ is the basic reproduction number
of the epidemic and $\cR_0 >1$, in  stochastic  epidemic models the probability of a major outbreak is some value in $(0,1)$
 \cite[7.5.1]{All}
 \cite[Chap 3]{DiHeBr} \cite[II.3]{Lud}. Hard-core
deterministic modelers may not like
such a random explanation and look for some deterministic
mechanism to explain this Swiss Covid-19 phenomenon.
 What comes to mind is some sort of Allee effect
\cite{Allee31, Allee51},
not for the host population
\cite[Sec.1.2.2]{IaMi} \cite[Chap.7]{Thi03} but for the  viruses \cite{REB}, that a host's chance of becoming infectious  depends
on the inhaled dose of infective particles. As far as time is concerned,  ``dose" can be understood in an  instantaneous
 \cite{LuSt-free, REB} or in a cumulative way
 \cite{Coo, HoWa1, HoWa2, Wal}.

Let $S(t)$ be the number of susceptible hosts at time $t$
and $I(t)$ the ``infective influence" \cite{Thi77} at time $t$,  in other words
 the ``expected pathogen exposition"
\cite{LuSt-free}. (See Section \ref{subsec:infective-influence}
for a detailed description.)
 In a Kermack--McKendrick type epidemic model,
which assumes mass-action infection kinetics (density-dependent incidence),
the differential equation
\begin{equation}
\label{eq:susc-intro}
S' = - SI
\end{equation}
is assumed. So the infective influence coincides with the
``force of infection"
\cite{BuCo, DiHeBr, IaMi, Inabook, Mar}.

\subsection{Instantaneous dose generalization}
\label{subsec:instant-dose}

An instantaneous dose generalization considers a force of infection $f(I)$ leading to the
more general differential equation
\begin{equation}
\label{eq:sus-nonlin}
S' =- S f(I)
\end{equation}
with a (not necessarily strictly) convex (or sigmoid \cite{REB}) function $f$, $f(0)=0$.
In \cite{LuSt-free}, $f$ is assumed to be piecewise linear,
i.e., $f$ is piecewise differentiable with piecewise constant
derivative $f'$. Alternatively,  $f(I) = I^p$
with $p > 1$ has been suggested in the literature  early on
(see \cite{All, Cap, DiIn, LHL, LLI, Nov, Sev}  and the references therein),
but this does not lead to a ``free boundary - in time" \cite{LuSt-free}.

\subsection{Cumulative dose generalization}
\label{subsec:cum-dose}

A cumulative  dose generalization
has been suggested by Ken Cooke in  1967 \cite{Coo}
and further analyzed in the 1970s \cite{Het, HoWa1, HoWa2, Wal, Wil} and later \cite{Smi, HoJa}.
It keeps (\ref{eq:susc-intro}) and assumes
that susceptible
hosts possess a natural resistance to (or tolerance of \cite{Wal}) the infectious
disease. ``A susceptible individual, $A$, does not become
infectious upon first exposure to an infectious individual, but only after repeated exposure to infectious individuals has broken down $A$'s resistance \cite{Coo}."
Differently from these references,
in the spirit of general Kermack/McKendrick models
\cite{KeMcK1, KeMcK2, KeMcK3},
our model will consider infective hosts whose infectivity
depends on their infectivity age, the time since becoming infective. ``Infectivity age" has already been incorporated in \cite{Thi77}, though  the term ``class age" has been used instead \cite{Hop, IaMi}.

Let $m > 0$ be the  value that describes
 the resistance to the disease and define the cumulative force of infection as
\begin{equation}
\label{eq:force-cum}
J(t) = \int_0^t I(s)ds, \qquad t \ge 0.
\end{equation}
We call $m$ the ``resistance" for short. It could also be called
``tolerance level", ``threshold of dose", ``threshold of dosage" or ``threshold of infection"  \cite{Wal}) but, to avoid confusion,  we want to mainly reserve "threshold"
for describing  the outbreak phenomena of the epidemic.
We assume the following:

\begin{itemize}
\item[(a)] If $t > 0$ and $J(t) = \int_0^t I(s)ds < m$, no susceptible host that has been
 infected after the start of the epidemic has turned infective by time $t$.

\item[(b)] If $t > 0$ and $J(t) >m$, then the set
\begin{equation}
\label{eq:thres-cum-force}
L(t,m) = \Big \{ r \in [0,t]; \int_r^t I(s) ds = m \Big \}
\end{equation}
is nonempty by the intermediate value theorem.
Here $\int_r^t I(s)ds$ represents the ``infective dosage" \cite{Coo}
an infected individual receives during the time from $r$ to $t$.
We assume that the hosts
that  turn
infective at time $t$ are exactly those that have been  infected at some time  $r \in L(t,m)$.
\end{itemize}

If $I$ is strictly positive, $L(t,m)$ is a singleton set;
otherwise it may be a nontrivial interval. But eventually, this will not
matter, because $J(s) = J(t)-m$ for all $s \in L(t,m)$.


\subsection{Preview of main results}
\label{subsec:preview}

Let $J_0 (t) \le J(t)$ be the cumulative force of infection
due to the initial infectives. Both $J_0$ and $J$ are increasing
functions of time and the following limits exist in $[0,\infty]$,
\begin{equation}
\label{eq:force-cum-final}
J_0(\infty) = \lim_{t\to \infty} J_0(t),
\qquad
J(\infty) =  \lim_{t\to \infty} J(t).
\end{equation}
Here $J(\infty)$ is the final cumulative force of infection
and $J_0(\infty)$ the final cumulative  force of infection due to the initial infectives.
Further, by (\ref{eq:susc-intro}),
\begin{equation}
S(t) = S(0) e^{-J(t)}, \qquad t \ge 0,
\end{equation}
and
\begin{equation}
\label{eq:final-size-sus}
S(\infty) = S(0) e^{-J(\infty)}
\end{equation}
 exists.

If $J_0(\infty) \le m$,  all  infective hosts
have already been infective at the beginning.

\begin{theorem}
\label{re:threshold-one-init-force}
If $J_0(\infty) \le m$, then $J(t) = J_0(t) $ for all $t \ge 0$ and $S(t) = S(0)e^{-J_0(t)}\ge S(0)e^{-m}$ for all $ t\ge 0$.
\end{theorem}

The next  result presents the other threshold side for $J_0(\infty)$
and the threshold character of the initial number of susceptible
hosts $S(0)$.

\begin{theorem}
\label{re:threshold-susc}
Let   $J_0(\infty) > m$. Then $J(\infty) > J_0(\infty) > m$
and $J(\infty)$ is a concave function of $J_0(\infty)\in (m,\infty)$.

Further, there exists some $S^\sharp >0$ (that depends on
other epidemiological parameters) with the following property:

\begin{itemize}

\item[(a)] Let $S(0) \le S^\sharp$. Then $J(\infty) \to m+$ as $J_0(\infty) \to m+$
and $S(\infty) \to S(0)e^{-m}$.

\item[(b)] Let $S(0) > S^\sharp$. Then, there exists some $x > 0$
(independent of $J_0$)
such that the following hold:

\begin{itemize}
\item[(i)]        If $J_0(\infty) > m$, then $J(\infty) \ge m + x$ and $S(\infty) \le S(0) e^{-(m+x)}$.

\item[(ii)] As $J_0(\infty) \to m+$, then $J(\infty) \to  m + x$ and $S(\infty) \to S(0) e^{-(m+x)}$.
\end{itemize}

\end{itemize}
\end{theorem}

So, if $S(0) > S^\sharp$, the final number of susceptible hosts jumps when
the final size of the cumulative initial force of infection
passes the  resistance $m$.

In early epidemic models, threshold theorems were often formulated
in terms of $S(0)$, the initial number of susceptible hosts \cite{Bai}, rather than the basic reproductive number $\cR_0$,
which became popular when the endemic situation of an
 infectious disease was considered.
 We will introduce $\cR_0$ when more terminology has been developed, (\ref{eq:bas-rep-no}),
and we will see that the sign of $S(0) - S^\sharp$ will be the
same as the sign of $\cR_0 -1$. In fact, the number $x$
 in part (b) of Theorem \ref{re:threshold-susc}
 is the unique positive solution to the scalar fixed point problem
\begin{equation}
x = \cR_0 (1-e^{-x}).
\end{equation}
See Section \ref{sec:final-size}.

As the bottom line, there is an easy possible explanation of the different
epidemic outcomes in the two Swiss cities:
While both cities had  basic reproduction numbers that exceeded
one, the final cumulative initial force of infection
could overcome the natural resistance to the disease
in one city but not in the other.

The derivation of the model that we will use follows \cite{Thi77},
but has been drastically simplified because here we do
not consider the spatial spread of an epidemic which was
the emphasis of \cite{Thi77} and its sequel \cite{Thi77b}.

In Sections \ref{sec:distributed} -- \ref{sec:dist-final-size}, we consider
a model in which susceptible hosts may have different
levels of resistance $m$ and study how our results
extend to this more general model.
For instance, if the  susceptible hosts have a continuous
resistance density that increases for small
levels of resistance  and if the number of initial
susceptibles is large enough, the final size
of the susceptible hosts still jumps  if
the final  cumulative initial  force of infection
  passes a certain threshold. The two thresholds for this are
 given in more complicated terms than the basic reproduction number
and a particular disease resistance (Theorem \ref{re:derivative}
and \ref{re:density-creates-jump} and equation
 (\ref{eq:threshold-jumps-exist-dist}) and
(\ref{eq:Gamma-superthr}).
 We also explore on which intervals the function
that relates the final cumulative force of infection to  the final cumulative force of initial infection is convex or concave.

In Section \ref{sec:example-Gamma}, we present an example
with  Gamma-distributed resistance in which jumps
occur if and only if the basic reproduction number
(without resistance) exceeds 4.46.
In fact, another threshold parameter becomes
relevant (Section \ref{subsec:new-thres}).
This
confirms the big effect that host heterogeneity
can have on epidemic outbreaks \cite{BoDI1}
and the importance of epidemic models with
population structure. See Figure
\ref{fig:Gamma-distrib} and  \ref{fig:Fixed-distrib},
which show the final cumulative force of infection
due to all infectives as a function of the
final cumulative force of infection due to
the initial infectives.
Figure \ref{fig:Fixed-distrib} shows this function
for a fixed resistance $m=3$ while Figure \ref{fig:Gamma-distrib} shows it for a
Gamma distributed resistance with mean resistance 3.


\section{The model for fixed resistance}
\label{sec:model}

In the introduction, we have spoken about the infective
influence which we will now describe in detail.

\subsection{Infective influence}
\label{subsec:infective-influence}

Let, at time $t \ge 0$, $u(t,a)$ be the  density of infective hosts with
infectivity age $a$, such that
\begin{equation}
\label{eq:infectives}
 \int_0^\infty u(t,a) da , \quad t \ge 0,
\end{equation}
is the number of infective hosts at time $t$. Infectivity age
is the time since becoming infective.

Notice that, as in \cite{Thi77} but differently from \cite{BoDI1, FCGTage, Inabook, Ina23, MaRu, Mar, PoTh}, we use infectivity age and not infection age
(time since infection)
in the derivation of our model.

The infective influence (which in our model is identical
to the force of infection) is
\begin{equation}
\label{eq:infectiv-influence}
I(t) = \int_0^\infty u(t,a)  \kappa(a) da.
\end{equation}
Here $\kappa(a)$ is the infective influence of an individual host
of infectivity age $a$.

The cumulative force of infection  is given by
\begin{equation}
\label{eq:cum-inf-infl}
J(t) = \int_0^t I(s) ds, \qquad t \ge 0.
\end{equation}
By (\ref{eq:infectiv-influence}) and (\ref{eq:cum-inf-infl}) and changing the order
of integration (Tonelli's theorem),
\begin{equation}
\label{eq:cum-inf-infl2}
J(t)
 =
\int_0^\infty \Big ( \int_0^t u(s,a) ds  \Big ) \kappa(a) da.
\end{equation}

Let $B(t)$ be the rate of hosts that become infective
at time $t$. So to speak, $B(t) $ is the
 birth rate of the disease and is called
``incidence". Let $P(a)$ be the probability that an infective
host has not died from the disease at infectivity age $a$, where $P:\R_+ \to [0,1]$ is a decreasing
function, $P(0)=1$.
Then
\begin{equation}
\label{eq:inf-age-dens}
\begin{split}
u(t,a) = & B(t-a) P(a), \qquad t > a \ge 0,
\\
u(t,a) = & u_0(a-t) \frac{P(a)}{P(a-t)}, \quad  a > t \ge 0,
\end{split}
\end{equation}
where $u_0$ is the  infectivity-age density of infective hosts  at time 0. $\frac{P(a)}{P(a-t)}$
is the conditional probability of not having died
from the disease at infectivity age $a$ provided
one has not died from it at infectivity age $a-t$.
Let
\begin{equation}
\label{eq:cum-infectivity-rate}
C(t) = \int_0^t B(s)ds, \qquad t \ge 0,
\end{equation}
 be the number of hosts
that have become infectious in the time interval from $0$
to $t$.
By (\ref{eq:inf-age-dens}),
\[
\int_a^t u(r,a) dr = \int_a^t B(r-a) P(a)dr = C(t-a) P(a), \quad  t > a >0,
\]
and
\[
\int_0^t u(r,a) dr= \int_0^t u_0(a-r) \frac{P(a)}{P(a-r)} dr
, \quad a \ge t\ge 0.
\]
In combination, for $t> a>0$,
\[
\begin{split}
\int_0^t u(r,a) dr = & \int_a^t u(r,a) dr + \int_0^a u(r,a)dr
\\
= &
C(t-a) P(a) + \int_0^a u_0(a-r) \frac{P(a)}{P(a-r)} dr.
\end{split}
\]
To obtain a formula that combines the last two
 equations for all $t \ge 0$,
we extend $C$, $u_0$, and $P$ from $\R_+$ to $\R$ by
\begin{equation}
\label{eq:extend}
C(t) :=0, \quad t < 0 , \qquad u_0(a) :=0,  \quad P(a):= P(0) =1, \quad a <0.
\end{equation}
Then
\begin{equation}
\label{eq:cum-inf-age-density}
\int_0^t u(r,a) da
=
C(t-a) P(a) + \int_0^t u_0(a-r) \frac{P(a)}{P(a-r)} dr, \quad t \ge 0.
\end{equation}
We substitute (\ref{eq:cum-inf-age-density})  into (\ref{eq:cum-inf-infl2}),
\begin{equation*}
J(t) = \int_0^\infty \Big (
C(t-a) P(a) + \int_0^t u_0(a-r) \frac{P(a)}{P(a-r)} dr
  \Big ) \kappa(a) da.
\end{equation*}
Since $C(t) =0$ for $t <0$ by (\ref{eq:extend}),
\begin{equation}
\label{eq:cum-inf-age-density2}
J(t) = \int_0^t
C(t-a) P(a) \kappa(a) da
+
J_0(t), \qquad t \ge 0,
\end{equation}
with
\begin{equation}
J_0(t)= \int_0^\infty
\Big ( \int_0^t u_0(a-r) \frac{P(a)}{P(a-r)} dr
  \Big ) \kappa(a) da, \qquad t \ge 0.
\end{equation}


\subsection{Dynamics of susceptibles}
\label{subsec:susceptibles}

Let $S(t)$ denote the number of susceptibles at time $t \ge 0$.
In an epidemic model, $-S'(t)$ is the disease incidence
(rate of new infections) at time $t$. $S(0) - S(t)$
is the prevalence of the disease (cumulative incidence)
of the disease at time $t$, the number of infections
from beginning to time $t$. We assume that infected individuals
do not become susceptible again. While this may not hold
for every infectious disease, the epidemic outbreak phase
may be essentially over before infective individuals become
susceptible again.
 So, we assume that $S'(t)$, the time derivative of $S$, equals
the rate of new infections,
\begin{equation}
\label{eq:susc-birth}
S'(t) = -B(t)
.
\end{equation}
Assuming density-dependent
(mass action type) incidence, we have
\begin{equation}
\label{eq:susceptibles-diff}
S'(t) = - S(t) I(t), \qquad t >0,
\end{equation}
where $I$ is the infective influence given by (\ref{eq:infectiv-influence}). So the infective
influence coincides with what is called {\em force of infection}
in parts of the literature \cite{BuCo, DiHeBr, Inabook, Mar}. We integrate (\ref{eq:susceptibles-diff}) and use (\ref{eq:cum-inf-infl}),
\begin{equation}
\label{eq:susceptibles}
S(t) = S(0) e^{-J(t)}, \quad t \ge 0,
\end{equation}
and
\begin{equation}
\label{eq:susceptibles-difference}
S(0) - S(t) = S(0) f(J(t)), \qquad t \ge 0,
\end{equation}
with
\begin{equation}
\label{eq:nonlinearity}
f(y) = 1 - e^{-y}, \quad y \ge 0.
\end{equation}
From (\ref{eq:susc-birth}) and  (\ref{eq:cum-infectivity-rate}) we also have that
\begin{equation}
S(0) - S(t) = C(t)
\end{equation}
is the number of hosts that have been infected in the time
interval $[0,t]$.


\subsection{A dose-dependent delay for becoming infective}
\label{subsec:threshold-infec}

Following \cite{Coo, HoWa1, HoWa2, Wal}, we assume that an infected individual
does not become necessarily and not immediately infective, but after a state-dependent  delay which depends on the force of infection $I$ in
(\ref{eq:infectiv-influence}). The underlying idea is that susceptible
hosts possess a natural resistance to the infectious
disease and only become infective (and possibly die from the
 disease) when the resistance
is overcome by the infectious influence of the epidemic in a cumulative way.

Let $m > 0$   describe the strength
of that resistance (see Section \ref{subsec:cum-dose}).

\begin{itemize}
\item If $t > 0$ and $J(t) = \int_0^t I(s)ds < m$, no host with
resistance $m$ that  was infected
after the start of the epidemic has turned infective by time $t$
and $C(t) =0$  for the cumulative number of  infected  individuals $C$ in (\ref{eq:cum-infectivity-rate}).

\item If $t > 0$ and $J(t) >m$, then the set
\begin{equation}
L(t,m) = \Big \{ r \in [0,t]; \int_r^t I(s) ds = m \Big \}
\end{equation}
is nonempty by the intermediate value theorem, and the hosts
that were infected at a time $r \in L(t,m)$  turn
infective at time $t$.
\end{itemize}

In the following, we elaborate the consideration in
\cite[p.343]{Thi77}.
We claim that, for any $t \ge 0$, a host becomes infective
in the time interval $[0,t]$ if and only if it was  infected
in an time interval $[0,r]$ with $r \in L(t,m)$.

Let
\[
\tau^\di(t,m) = \sup L(t,m) \hbox{ and }
\tau_\di(t,m) = \inf L(t,m).
\]
Notice that
\[
\tau^\di(t,m), \;\tau_\di(t,m) \in L(t,m).
\]
Let $t \ge 0$ and $\int_0^t I(s)ds \ge m$.

Consider a host that  becomes infective in the interval $[0,t]$.
Then it has become infective at some time $\tilde t \in [0,t]$
and was  infected at some time $\tilde r $ with $\tilde r \in L(\tilde t, m)$, $ \int_{\tilde r}^{\tilde t} I(s) ds =m$.
Then $\int_{\tilde r}^t I(s) ds \ge m$ and $\int_r^t I(s)ds
=m $ for some $r \ge \tilde r$. So the host was
 infected
in some interval $[0,r]$ with $r \in L(t,m)$ and in the interval
$[0, \tau^\di(t,m)]$.

Conversely, consider a host that was  infected in the interval
$[0, \tau_\di(t,m)$].

Then it was infected in an interval $[0,r]$
with $r \in L(t,m)$,  $\int_r^t I(s) ds =m$. So, it was  infected at some time $\tilde r \in [0,r]$. Then $\int_{\tilde r}^t I(s) ds \ge m$, and this host
 becomes infective in the interval $[0,t]$.

Since $S$ decreases and $\tau_\di(t,m) \le \tau^\di(t,m)$,
\[
S(0) - S(\tau_\di(t,m)) \le C(t) \le S(0) - S(\tau^\di(t,m)).
\]
For $r = \tau_\di(t,m), \tau^\di(t,m)$,
\[
S(0) - S(r) = S(0) f(J(r))= S(0) f(J(t)-m),
\]
and so
\[
C(t) = f(J(t)-m).
\]
If $J(t) < m$, $C(t,m) =0 $ and, since $f(0) =0$, $C(t,m) = S(0,m) f(0)=0$.
In combination,
\begin{equation}
\label{eq:cum-infec-rate-func}
C(t) = S(0)  f \big ([J(t) - m]_+ \big )  , \qquad t \ge 0,
\end{equation}
where $[r]_+ = \max \{r,0\}$ is the positive part of the number
$r \in \R$. Cf. \cite[(8)]{Thi77}.


\subsection{The model equation}
\label{subsec:mod-integ-eq}

In this section, we will describe the course of the epidemic by
a single integral equation that essentially is of Volterra type.
This integral equation is in terms of the cumulative infectious
influence $J$ (cumulative force of infection). By (\ref{eq:susceptibles}), $J$ can be expressed in more epidemiologically
tangible terms via the susceptible hosts by
\begin{equation}
\label{eq:cum-inf-infl--suscep}
J(t) =  \ln \frac{S(0)}{S(t)}, \qquad t \ge 0.
\end{equation}
  Recall (\ref{eq:cum-inf-age-density2}),
\begin{equation}
J(t) = \int_0^t
C(t-a) P(a) \kappa(a) da
+
J_0(t)
\end{equation}
with
\begin{equation}
J_0(t)= \int_0^\infty
\Big ( \int_0^t u_0(a-r) \frac{P(a)}{P(a-r)} dr
  \Big ) \kappa(a) da.
\end{equation}
By (\ref{eq:cum-infec-rate-func}),
\begin{equation}
\label{eq:model-Volt-int}
J(t) = \int_0^t  S(0) f \big ([J(t-a) - m]_+ \big ) P(a) \kappa(a) da
+
J_0(t).
\end{equation}
Here, after changing the order of integration and
the variables and using $u_0(s) =0$ for $s <0$,
\begin{equation}
\label{eq:cum-init-infec-influ}
J_0(t) = \int_0^t \Big ( \int_0^\infty u_0(s) \frac{P(s+r)}{P(s)}
\kappa(s+r) ds \Big ) dr.
\end{equation}
where $J_0$ can be interpreted as the cumulative initial force of infection.

Assume that $\kappa$ is essentially bounded and $u_0$
is finitely integrable. Since $P$ is decreasing,
\[
 \int_0^\infty u_0(s) \frac{P(s+r)}{P(s)}
\kappa(s+r) ds \le \kappa^\di \int_0^\infty u_0 ds
, \qquad r \ge 0,
\]
where $\kappa^\di$ is the essential supremum of $\kappa$,
and $J_0: \R_+ \to \R_+$ is continuous and increasing.



\subsection{Existence of solutions}
\label{sec:existence}

Consider $X =C(\R_+)$ and $X_+ = C_+(\R_+)$
where $X$ is the vector space of continuous real-valued
functions and $X_+$ the cone of nonnegative functions in $X$.
Define $G:X_+ \to X_+$ by
\begin{equation}
\label{eq:map}
G(J)(t)  =  \int_0^t   f \big ([J(t-a) - m]_+ \big)
 \xi(a) da
+
J_0(t), \qquad J \in X_+, \quad t \in \R_+,
\end{equation}
with
\begin{equation}
\label{eq:kernel}
\xi(a) = S(0) P(a) \kappa(a).
\end{equation}

Since $f: \R_+ \to \R_+$ given by (\ref{eq:nonlinearity}) is increasing and $f(\R_+) =[0,1)$, $G$ maps increasing functions to increasing functions
and $G$ is increasing itself:

\begin{lemma}
\label{re:map}
(a) If $J(t) \le \tilde J(t)$ for all $t \ge 0$,
then $G(J)(t) \le G(\tilde J)(t)$ for all $t \ge 0$.

(b) If $J:\R_+ \to \R_+$ is increasing, $G(f):\R_+ \to \R_+$ is increasing.

(c) For all $J \in X_+$,
\[
G(J) (t) \le \int_0^t \xi(a) da + J_0(t), \qquad t \ge 0.
\]
\end{lemma}

Define a sequence of functions $J_n: X_+ \to X_+$ by
$J_n = G(J_{n-1}) $ for all $n \in \N$.

Notice that $J_n(t) \ge J_{n-1}(t)$ for all $t \ge 0$
first for $n =1$ and then for all $n \in \N$ by induction.
By induction again and Lemma \ref{re:map}, $J_n(t)$ is increasing in $t \ge 0$ and $n \in \Z_+$. Further, the
sequences $(J_n(t))_{n\in \N}$ are bounded for all $t \ge 0$.

By the monotone convergence theorem,
$J(t) = \lim_{n\to \infty} J_n(t)$ exists for all $t \ge 0$,
defines a Borel measurable function on $\R_+$ and satisfies
$J =G(J)$, i.e., it is a solution of (\ref{eq:model-Volt-int}).
Restricting $G$ to $C_+[0, T])$ with $T \in (0,\infty)$
and using the contraction mapping theorem, one obtains
that $J$ is the unique solution and continuous.

\begin{corollary}
\label{re:threshold-one}
Assume that $J_0(t) \le m$ for all $t \ge 0$. Then
$J(t) = J_0(t)$ for all $t \ge 0$ for the unique solution
$J$ of (\ref{eq:model-Volt-int}).
\end{corollary}

\begin{proof}
By induction and (\ref{eq:map}), $J_n(t) = J_0(t)\le m$ for
all $n \in \N$ and $t \ge 0$. So, also for the limit,
$J(t) = J_0(t) $ for all $t \ge 0$.
\end{proof}


\section{The final size of the epidemic for fixed resistance}
\label{sec:final-size}

We consider equation (\ref{eq:model-Volt-int})
for the cumulative force of infection and the
behavior of its solutions as time tends to infinity.

\subsection{The final size of the cumulative initial force of  infection}

Consistent with its interpretation,
the cumulative initial force of infection $J_0$ in (\ref{eq:cum-init-infec-influ}) is an increasing function
of time, which has a limit in $[0,\infty]$ as $t \to \infty$,
\begin{equation}
\label{eq:cum-init-infec-influ-fin}
J_0(t) \to  J_0^\infty = \int_0^\infty \Big ( \int_0^\infty u_0(s) \frac{P(s+r)}{P(s)}
\kappa(s+r) ds \Big ) dr.
\end{equation}
See \cite[Sec.3.2]{FCGTage} for upper and lower estimates
of $J_0^\infty$.


\subsection{The final size of the cumulative force of infection}


Since $J$ is an increasing function of time,
\begin{equation}
J^\infty = \lim_{t\to \infty} J(t)
\end{equation}
exists. By Beppo Levi's theorem of monotone convergence,
applied to (\ref{eq:model-Volt-int}),
\begin{equation}
\label{eq:model-Volt-int-final}
J^\infty =  \cR_0  f ([J^\infty - m]_+)
+
J_0^\infty.
\end{equation}
Here
\begin{equation}
\label{eq:bas-rep-no}
\cR_0 = S(0) \int_0^\infty P(a) \kappa(a) da
\end{equation}
is the basic reproduction number of the epidemic
(without resistance) and has the following interpretation:
\begin{definition}
\label{def:R0interpret}
$\cR_0$ is the average number of infections generated
by a typical freshly infective individual
during its entire period of infectivity if it
is introduced into a
 host population of size $S(0)$ that is completely susceptible
with no host being resistant to the disease (in the
sense explained in Section \ref{subsec:cum-dose}).
\end{definition}

$J^\infty$ can also be obtained as the limit of
the following recursive sequence:
\begin{equation}
\begin{split}
z_0 = & J_0^\infty,
\\
z_n = & \cR_0 f ([z_{n-1} - m]_+) + z_0, \qquad n \in \N.
\end{split}
\end{equation}
By induction, since $f$ is increasing, $(z_n)$ is an increasing sequence which is
bounded by $\cR_0 +z_0$ and thus has a limit $z$ which solves
$z = \cR_0 f([z-m]) + z_0$. Further, $z$ is the minimal
solution of this equation. Let $\tilde z$ be another solution.
By induction, $z_n \le \tilde z$ for all $n \in \N$ and so $z \le
\tilde z$.

Consider the recursion $J_n = G(J_{n-1})$ with $G$ in (\ref{eq:map}). By induction,
\[
J_n(t) \le z_n, \qquad t \ge 0, \quad n \in \N.
\]
In the limit, $J(t) \le z $ for all $t \ge 0$
and so $J^\infty \le z $. Since  both $z$ and $J^\infty$
are solutions of (\ref{eq:model-Volt-int-final})
and $z$ is the minimal solution, $J^\infty =z$.

Compare \cite[Sec.3]{Thi77}.

Here are the double threshold results. The first
has already been essentially  stated in Corollary \ref{re:threshold-one} and elaborates one side of the character of $m$
as threshold parameter for the epidemic. Also use (\ref{eq:susceptibles}).

\begin{theorem}
\label{re:threshold-one-elab}
If $J_0^\infty \le m$, then $J(t) = J_0(t) $ for all $t \ge 0$ and $S(t) = S(0)e^{-J_0(t)}\ge S(0)e^{-m}$ for all $ t\ge 0$.
\end{theorem}

The second results presents the other threshold side for $m$
and the threshold character of the basic reproduction
number $\cR_0$.

\begin{theorem}
\label{re:threshold-two-elab}
Let   $J_0^\infty > m$. Then $J^\infty > m$ and $J^\infty$
is the unique solution $z>m$ of $z = \cR_0 f(z-m) +z_0$ with $z_0 =J_0^\infty$. Further, $J^\infty$ is a concave and  continuous
function  of $J_0^\infty >m$.

\begin{itemize}
\item[(a)] As $J_0^\infty \to \infty$, $J^\infty \to \infty$
and $J^\infty - J_0^\infty \to \cR_0$.

\item[(b)] Let $\cR_0 \le 1$. Then $J^\infty \to m$ as $J_0^\infty \to m$
and $S^\infty \to S(0)e^{-m}$.

\item[(c)] Let $\cR_0 >1$. Then, there exists a unique solution $x > 0$
of the equation $x = \cR_0 f(x)$.

\begin{itemize}
\item[(i)] Whenever $J_0^\infty > m$, then $J^\infty \ge m + x$ and $S^\infty \le S(0) e^{-(m+x)}=:S^\di$.

 \item[(ii)] As $J_0^\infty \to m+$, $J^\infty \to m + x$ and $S^\infty \to S(0) e^{-(m+x)}$.
\end{itemize}
\end{itemize}
\end{theorem}

\begin{remark}
If $\cR_0 > 1$, and $0< x =f (x)$, then the equation $z = f([z-m]_+) +m$ has the solutions $z =m$ and $z=m + x$, but $J^\infty =m$
is the epidemiologically relevant solution for $J_0^\infty=m$.

Let $\cR_0 >1 $ and $J_0^\infty < m$.
 Then $J^\infty = J_0^\infty$
is the epidemiologically relevant solution of
(\ref{eq:model-Volt-int-final}) and is the only solution
of this equation if $\cR_0 \le 1$. If $\cR_0 > 1$ is large
enough, there is  also a solution $y > m$.

Set $y_1 = m+1$. Then
\[
y_1 \le \cR_0 f([y_1-m]_+) + J_0^\infty
\quad \hbox{if} \quad \cR_0 f(1) +J_0^\infty > m+1.
\]
 Further,
\[
y_2 \ge  \cR_0 f([y_2-m]_+) + J_0^\infty \quad
 \hbox{if} \quad y_2 > \cR_0 + J_0^\infty.
\]
 By the intermediate value theorem, there is a solution between $y_1 $ and $y_2$.
\end{remark}

The formulation in Theorem \ref{re:threshold-susc}
is obtained from Theorem \ref{re:threshold-two-elab}
by setting
\begin{equation}
\label{eq:susc-thresh}
S^\sharp = \Big ( \int_0^\infty P(a) \kappa(a) da \Big )^{-1}.
\end{equation}
By (\ref{eq:bas-rep-no}), $\cR_0 -1 $ and $S(0) - S^\sharp$
have the same sign.

\begin{proof}[Proof of Theorem \ref{re:threshold-two-elab}]

The first statement holds because $J_0(\infty) \le J(\infty)$.

Define $y = J(\infty) -m $ and $y_0 = J_0(\infty) -m$. Then $y \ge y_0> 0$ and, by (\ref{eq:model-Volt-int-final}),
\begin{equation}
\label{eq:model-Volt-int-final-shift}
y =  \cR_0 f (y)
+
y_0,
\end{equation}
and the analysis reduces to the well-known problem with $m=0$.
By (\ref{eq:nonlinearity}), $f$ is strictly increasing and strictly concave, $f(0)=1$, $f(\infty) =1$, and infinitely often differentiable,
$f'(0) =1$, and $f(y)/y$ is strictly decreasing.

In particular, the solution $y$ is uniquely determined by $y_0$.
Since this plays out in the real numbers, uniqueness implies
continuous dependence.

Uniqueness follows from the fact that $f(y) /y$ is a strictly
decreasing function of $y>0$.

That $J^\infty$ is a concave function of $J_0^\infty \ge m$
follows from Theorem \ref{re:convex-inherit}
and the fact that $F(y) = \cR_0 f([y-m]_+) + J_0^\infty$
defines a concave function $F$ of $y \ge m$.

(a) Notice that $f(z) \to 1$ as $z \to \infty$ by (\ref{eq:nonlinearity}).

(b) If $\cR_0 \le 1$ and $y$ is a solution to (\ref{eq:model-Volt-int-final-shift}), $y \to 0$ as $y_0 \to 0$.

(c) Let $\cR_0 > 1$. Then, there exists a  solution
$x > 0$ of $1 = \cR_0 f(x)/x$ by the intermediate value theorem
which is unique. Further, if $y$ is a solution to (\ref{eq:model-Volt-int-final-shift}), $y >x $ for all $y_0 >0$ and $y \to x$
as $y_0 \to 0$.
\end{proof}



\section{The model for distributed  resistance}
\label{sec:distributed}

It is very unlikely that all susceptibles have the same
resistance to the disease though we assume
that they all have some resistance. So we assume that, at time $t \ge 0$,
the number of susceptibles is given by $\int_0^\infty S(t,m) dm $,
where, as a function of $m$,  $S(t,m) $ is the density
of susceptible hosts with  resistance $m \ge 0$.
We still assume that all susceptible hosts have the same chance
of being exposed to the infection.
For models with distributed susceptibility
or distributed other traits,
see \cite{BoDI1, BoDI2} and the references therein.
We also make
the simplifying assumption that
the level of resistance does not affect the disease mortality
of an infective host or its infectivity.
See the discussion.

At time $t \ge 0$, let $u_1(t,a)$ be the  density of  primary infective
hosts with infectivity age $a \ge t$,
\begin{equation}
u_1(t,a) = u_0(a-t) \frac{P(a-t)}{P(a)}, \qquad a \ge t \ge 0,
\end{equation}
where $u_0(a)$ is the  density of initially infective hosts
with infectivity age $a$
and $u_0: \R_+ \to \R_+$ is a given finitely integrable function.
\[
\int_t^\infty u_1(t,a) da
\]
is the number of primary infective hosts at time $t$.
Notice that all primary infective hosts at time $t$  have
an infectivity age $a \ge t$.

At time $t \ge 0$, let  $u_2(t,a,m)$ be the density of secondary infective hosts with
infectivity age $a \in [0,t]$ and resistance  $m$, such that
\begin{equation}
\label{eq:infectives-dis}
\int_0^\infty  \int_0^t u_2(t,a,m ) \; da \, dm , \quad t \ge 0,
\end{equation}
is the number of secondary infective individuals at time $t$. Infectivity age
is the time since becoming infective. Notice that all secondary
infective hosts at time $t$ have an infectivity age $a \in [0,t)$.

The force of infection is now
\begin{equation}
\label{eq:FoI-dist}
I(t) = \int_t^\infty u_1(t,a)  \kappa(a) da
+
\int_0^t \Big ( \int_0^\infty u_2(t,a,m) dm \Big )  \kappa(a) da.
\end{equation}
Let $B(t,m)$ be the rate of individuals with resistance  $m$ that become infective
at time $t$. Let $P(a)$ be the probability that an infective
individual has not died from the disease at age $a$.
Then
\begin{equation}
\label{eq:sec-inf-age-dens}
u_2(t,a,m) =  B(t-a,m) P(a), \quad t > a \ge 0.
\end{equation}
The equation for the susceptible hosts becomes
\begin{equation}
\p_t S(t,m) =- S(t,m) I(t),
\end{equation}
with $\p_t $ denoting the partial derivative with respect to time $t$.
It is solved by
\begin{equation}
\label{eq:sus-cumforce-dist}
S(t,m) = S(0,m) e^{-J(t)}, \qquad J(t) = \int_0^t I(s) ds.
\end{equation}
Let, for $t\ge 0$,
\begin{equation}
\label{eq:cum-infectivity-rate-dis}
C(t,m) = \int_0^t B(s,m)ds, \qquad  m \ge 0,
\end{equation}
 be the  density of hosts with resistance  $m$
that have become infectious in the time interval from $0$
to $t$. The same reasoning as in Section  \ref{subsec:threshold-infec}  provides
\begin{equation}
\label{eq:cum-infec-rate-func-dis}
C(t,m) = S(0,m)  f ([J(t) - m]_+)  , \qquad t,m \ge 0,
\end{equation}
where $[r]_+ = \max \{r,0\}$ is the positive part of the number
$r \in \R$.
The same reasoning as in Section
\ref{subsec:infective-influence} leads to
\begin{equation}
\label{eq:model-Volt-int-dist}
J(t) = \int_0^t \Big (\int_{0}^{\infty} S(0, m) f ([J(t-a) - m]_+) dm \Big)
P(a)  \kappa (a) da
+
J_0(t),
\end{equation}
with $J_0$ from (\ref{eq:cum-init-infec-influ}).

At this point, for consistency with Section \ref{sec:model},
we introduce the number  of the initial susceptibles
\begin{equation}
\label{eq:susc-dist-tot}
S_0 = \int_0^\infty S(0,m) dm
\end{equation}
and the probability density of their resistance distribution
\begin{equation}
\sigma(m) =  S(0,m)/S_0.
\end{equation}
Let
\begin{equation}
\label{eq:xi}
\xi (a) = S_0 P(a ) \kappa(a).
\end{equation}
Then
\begin{equation}
\label{eq:model-Volt-int-dist1}
\begin{split}
J(t) = & \int_0^t \Big (\int_{0}^{\infty} \sigma(m) f ([J(t-a) - m]_+) dm \Big)
\xi(a) da
\\ & +
J_0(t), \qquad t \in \R_+.
\end{split}
\end{equation}

The increasing solution $J: \R_+ \to \R_+$ can be found as the pointwise
limit of an increasing iteration with start function $J_0$
as in Section \ref{sec:existence}.

\begin{proposition}
\label{re:minimal}
If  $\int_0^{J_0(\infty)} \sigma(m) dm  =0$,
then $J(t) = J_0(t)$ for all $t \ge 0$.
\end{proposition}

\begin{proof}
Notice that
\[
\int_0^t \Big (\int_{0}^{\infty} \sigma(m) f ([J_0(t-a) - m]_+) dm \Big)
\xi(a) da
=0, \qquad t \ge 0. \qedhere
\]
\end{proof}

Let $t \to \infty$ in (\ref{eq:model-Volt-int-dist1}),
to obtain the final size of the cumulative
infective force,
\begin{equation}
\label{eq:model-Volt-int0}
J(\infty) = \cR_\di \int_0^{\infty}   \sigma(m) f \big ([J(\infty) - m]_+\big ) dm
  +
J_0(\infty).
\end{equation}
Here
\begin{equation}
\label{eq:Xi}
\cR_\di =  \int_0^\infty \xi(a) da = S_0 \int_0^\infty P(a) \kappa(a) da
\end{equation}
is the basic reproduction number as interpreted
in Definition \ref{def:R0interpret}.
We use a new symbol, though, because $\cR_\di$
has no longer the threshold property that $\cR_0$
had before.

Since $f(0)=0$,
\begin{equation}
\label{eq:model-Volt-int-dis}
J(\infty) = \cR_\di \int_0^{J(\infty)}  \sigma(m) f (J(\infty) - m) dm +
J_0(\infty).
\end{equation}
By (\ref{eq:sus-cumforce-dist}),
\begin{equation}
\label{eq:sus-cumforce-fin-dist}
S_\infty = S_0 e^{-J(\infty)}, \qquad J(\infty) = \int_0^\infty I(s) ds,
\end{equation}
where $S_\infty$ is the final size of the susceptible part of the host population,
\[
S_\infty = \lim_{t\to \infty}  \int_0^\infty S(t,m)dm.
\]
and $S_0$ is the initial size, (\ref{eq:susc-dist-tot}).

As in  Section \ref{sec:final-size},
$J(\infty) = \lim_{t\to \infty} J(t)$
 is the minimal solution of this equation
and can be found as the limit of an increasing
iteration starting with $J_0(\infty)$.
Compare \cite[Sec.3]{Thi77}.

\begin{proposition}
\label{re:estim-below}
Let
\begin{equation}
\label{eq:estim-below1}
\cR_\di \int_0^{J_0(\infty)}  \sigma( m)  dm  >1.
\end{equation}
Then $J(\infty) \ge J_0(\infty) +x$ where $x>0$ is the unique
solution of
\begin{equation}
\label{eq:extra-fixed-dist}
x= f(x)\cR_\di \int_0^{J_0(\infty)} \sigma(m)dm  .
\end{equation}
\end{proposition}

\begin{proof}
Recall (\ref{eq:model-Volt-int-dis}).
Set $y  = J(\infty) - J_0(\infty)$. By (\ref{eq:estim-below1}),
since $f(x) > 0$ for $x >0$,
\begin{equation}
\label{eq:estim-below2}
\int_0^{J_0(\infty)} \sigma(m) f(J_0(\infty) -m)  dm >0.
\end{equation}
and
\begin{equation}
\label{eq:model-Volt-int-mod}
 y= \cR_\di \int_0^{J(\infty)}  \sigma( m)
 f \big (y + J_0(\infty) - m \big ) dm
.
\end{equation}
By \ref{eq:estim-below2}, $y > 0$ and
\begin{equation}
\label{eq:model-Volt-int-ineq}
 y\ge f (y ) \cR_\di \int_0^{J_0(\infty)}  \sigma( m)  dm
.
\end{equation}
Since $f(y)/y$ is a strictly decreasing function of $y>0$, $y \ge  x$ for the unique solution $x$ of \ref{eq:extra-fixed-dist}
which exists by the intermediate value theorem and (\ref{eq:estim-below1}).
\end{proof}


\section{Analysis of an abstract final size equation}
\label{sec:abst-fin-size}

Let $x \in \R_+$ denote the final cumulative force
of infection due to the initial infectives
and $z(x) \in \R_+$ the final cumulative force of
infection due to all infectives. In terms of the previous sections, $x= J_0^\infty$ and $z(x) =
J^\infty$. We have found the following equation,
\begin{equation}
\label{eq:final-size-function}
z(x) = x + F(z(x)), \qquad x \in \R_+.
\end{equation}
Here,  $F: \R_+ \to \R_+$ is continuous, bounded and increasing, $F(0) =0$, with the finite limit
\[
F(\infty) := \lim_{y \to \infty} F(y).
\]
We want to find information about the the shape
of the function $z: \R_+ \to \R_+$.

For fixed resistance $m$, we have
\begin{equation}
\label{eq:F-fixed}
F(y) = \cR_0 f([y-m]_+), \qquad y \in \R_+.
\end{equation}
See Section \ref{sec:final-size}.
For continuously distributed resistance, we employ
\begin{equation}
\label{eq:F-dist}
F(y) = \cR_\di \int_0^y \sigma(m) f(y-m) dm, \qquad y \in \R_+.
\end{equation}
See Section \ref{sec:distributed}. In both cases, $f$ is the Ivlev-Skellam function
\begin{equation}
\label{eq:skellam}
f(y) =1 -e^{-y}, \qquad y \in \R_+,
\end{equation}
and $\cR_0$ and $\cR_\di$ are the basic reproduction number
(without resistance)
given by the identical formulas in (\ref{eq:bas-rep-no}) and (\ref{eq:Xi}).
 In the second case, $\sigma:\R_+ \to \R_+$ is continuous and
its integral is 1, $\cR_\di \in (0,\infty)$.
 Different symbols, $\cR_0$ and $\cR_\di$, are
used because $\cR_0$ will act as a threshold
parameter as expected from a basic reproduction number while $\cR_\di$ will not.

\begin{remark}
\label{re:F-R0}
In both cases,
\begin{equation}
\label{eq:F-dist-R0}
F(y) = \cR_\di \tilde F(y), \qquad y \in \R_+,
\end{equation}
where $\cR_\di = \cR_0 \in (0,\infty)$ is the basic reproduction number
of the disease given by (\ref{eq:Xi})
(it there were no resistance:
Definition \ref{def:R0interpret})
and $\tilde F: \R \to \R$ is bounded, increasing
and continuous and $\tilde F(\infty) =1$.
\end{remark}

In the spirit of \cite{BoDI1},
 we could (though we will not) also consider
\begin{equation}
\label{eq:measure}
F(y) = \int_{\R_+} f([y-m]_+) \mu(dm), \quad y \in \R_+,
\end{equation}
with a finite nonnegative Borel measure $\mu$.
In that case, continuity of $F$ follows from Lebesgue's
theorem of dominated convergence.

To find the solution function $z$ of (\ref{eq:final-size-function}), let $(z_n)$ be a  sequence of functions $z_n:\R_+ \to \R_+$ given recursively by
\begin{equation}
\label{eq:recurs-functions-app}
\left.
\begin{array}{rl}
z_0 (x) =& x
\\
z_n(x) = & F(z_{n-1}(x)) + x , \quad n \in \N
\end{array}
\right \}
\qquad
x \in \R_+.
\end{equation}
It follows by induction that $z_n(x)$ is increasing in both $n \in \Z_+$ and $x \in \R_+$, $z_n(0)=0$. Since $F$ is bounded on $\R_+$, for each $x
\in \R_+$, $(z_n(x))$ is a bounded increasing sequence in $\R_+$
and the limits
\begin{equation}
\label{eq:recursion-limit-app}
z(x) = \lim_{n\to \infty} z_n(x), \qquad x \in \R_+,
\end{equation}
exist and provide an increasing function $z:\R_+ \to \R_+$,
\begin{equation}
z(0)=0, \qquad 0 \le z(x)-x  \le F(\infty), \quad x \in \R.
\end{equation}

In general, the convergence is only pointwise and $z$ is
not continuous. See Theorem \ref{re:density-creates-jump}.
Since $F$ is continuous, we can take limits in
(\ref{eq:recurs-functions-app}),
\begin{equation}
\label{eq:recursion-limit-equation-app}
z(x) = F(z(x)) + x , \quad x \in \R_+.
\end{equation}
It follows from this equation that $z$ is strictly increasing.

\begin{theorem}
\label{re:minimal-sol-app}
For each $x \in \R_+$, $z(x)$ is the minimal solution
of the equation $\tilde z = F(\tilde z) +x$, i.e.,
$z(x) \le \tilde z$ for any solution $\tilde z$ of this
equation.

Further, $z$ is continuous from the left.

For each $x \in \R_+$, $z(x) =x$ if and only if $F(x) =0$.

Finally, $z(x) -x$ is an increasing function of $x \in \R_+$, $z$
is a strictly increasing function, and
\[
0 \le z(x) - x \to F(\infty), \qquad x \to \infty.
\]
\end{theorem}

In the epidemic context, we call $z:\R_+ \to \R_+$
the final size curve. It gives the final size
of the cumulative force of infection
as a function of the final size of the cumulative
initial force of infection. In general, there
may be several solution of (\ref{eq:final-size-function}). The minimal solution is the epidemiologically relevant one because it
also is the limit of the cumulative force of infection
as time tends to infinity (as we have shown
in the previous sections). Compare \cite[Sec.3]{Thi77}.

\begin{proof}
Let $x \in \R_+$. By induction, $x \le z_n(x) \le \tilde z$ for any solution
$\tilde z$  where $z_n$ is from the recursion
(\ref{eq:recurs-functions-app}).

By taking the limit as $n \to \infty$,
 (\ref{eq:recursion-limit-equation-app}),
$x \le z(x) \le \tilde z$.

Let $F(x) =0$. Then $x$ is a solution of $x = F(x) +x $.
Since $z(x)$ is the minimal solution, $z(x) =x$.
Further, since $F$ is increasing, $z(x) \ge F(x) + x$.
So, if $z(x) =x$, $F(x) =0$.

Suppose that $z$ is not continuous from the left at $x$.
Then there exists some $\epsilon > 0$ and a sequence
$(x_n)$ with $x_n < x $ for all $n \in \N$ and $x_n \to x$
as $n \to \infty$ and $|z(x_n) - z(x)| > \epsilon$
for all $n \in \N$. Since $z$ is increasing, $z(x_n) < z(x)-\epsilon$ for all $n \in \N$.  After choosing
a subsequence, there is some $\tilde z $ such that $z(x_n) \to \tilde z$ for $n \to \infty$,
$\tilde z \le z(x) -\epsilon$.
Since $F$ is continuous, by (\ref{eq:recursion-limit-equation-app}),
$\tilde z = F(\tilde z) + x$. Since $z(x)$ is the minimal
solution of this equation, we obtain the contradiction $z(x)\le \tilde z$.

Finally $z(x) \to \infty$ as $x \to \infty$ in an increasing
fashion and $z(x) -x = F(z(x)) \to F(\infty)$ as $x \to \infty$.
\end{proof}

\begin{remark}
\label{re:right-inverse}
Let the function $E: \R_+ \to \R$ be defined by
\begin{equation}
\label{eq:functionE}
E(x) = x -F(x), \qquad x \in \R_+.
\end{equation}
  Let $z$ be the minimal solution of (\ref{eq:recursion-limit-equation-app}). Then
\begin{equation}
\label{eq:z-right-inverse}
x = E(z(x)), \qquad x\in \R_+.
\end{equation}
Actually,
\begin{equation}
\label{eq:z-right-inverse-min}
z(x) = \min E^{-1}(x), \quad x \in \R_+,
\end{equation}
with $E^{-1}(x) =\{z \in \R_+; E(z)=x \}$.
Equation (\ref{eq:z-right-inverse}) implies that $E(\R_+)\supseteq \R_+$.
\end{remark}


\subsection{A guiding result}
\label{subsec:guide}

 The next result will serve as a road map for exploring jumps of the final size curve.
It gives necessary and sufficient conditions
for the non-existence of jumps.

\begin{theorem}
\label{re:equivalences}
Let $E:\R_+ \to \R$ be the function defined in (\ref{eq:functionE}) and $z:\R_+\to \R_+$
be the final size curve (\ref{eq:z-right-inverse}).

Then the following are equivalent:

\begin{itemize}
\item[(a)] $E$ is strictly increasing.

\item[(b)] $z$ is surjective.

\item[(c)] $z$ is continuous.
\end{itemize}
\end{theorem}

\begin{proof}
Assume (a). Since $E(0) =0$, $E(\R_+) \subseteq \R_+$ and $E(\R_+) = \R_+$ by Remark \ref{re:right-inverse}.
Suppose that $z$ is not continuous.
Then there exists some $x \in \R_+$ and
a sequence $(x_n)$ in $\R$ such that $x_n
\to x$ as $n \to \infty$, but $(z(x_n))$
does not converge to $z(x)$ as $n \to \infty$.
Since $z$ is increasing, the sequence $(z(x_n))$
is bounded. After choosing a subsequence, still  $
x_n \to x$, but $z(x_n) \to \tilde z \ne z(x)$
as $n \to \infty$. Since $E$ is continuous
and strictly increasing, $E(z(x_n)) \to E(\tilde z)
\ne E(z(x))$. By (\ref{eq:z-right-inverse}),
$x_n \to E(\tilde z) \ne x$, a contradiction.

Thus $z$ is surjective and continuous, i.e., (a) and (c) hold.

Assume (b). Let $0 \le z_1< z_2 < \infty$.
Since $z$ is surjective, there exist
$x_1,x_2 \in \R_+$ such that $z(x_i) = z_i$
for $j =1,2$. Since $z$ is strictly increasing,
$x_1 < x_2$. By  (\ref{eq:z-right-inverse}),
$E(z_1) =x_1 < x_2 = E(z_2)$. Thus, $E$ is
strictly increasing and (a) holds.

Assume (c). Since $z(0)=0$ and $z(x) \to \infty$
as $x \to \infty$, $z$ is surjective as a
consequence of the intermediate value theorem
and (b) holds.
\end{proof}

\begin{corollary}
\label{re:jump-equiv}
Let $E:\R_+ \to \R$ be the function defined in (\ref{eq:functionE}) and $z:\R_+\to \R_+$
be the final size curve (\ref{eq:z-right-inverse}).
Then $z$ makes a jump if and only if
$E$ is not strictly increasing.

If $E$ is strictly increasing, $z$ is the inverse
function of $E$.
\end{corollary}


\subsection{Fixed resistance revisited}
\label{exp:fixed}
Let us return to fixed resistance $m$ where
$F$ is given by (\ref{eq:F-fixed}),
\[
E(y) = y -\cR_0 \big (1 -e^{-[y-m]_+}\big), \qquad y \in \R_+.
\]
Then $E$ is differentiable except at $m$,
\[
\begin{array}{ll}
 E'(y) = 1, & 0\le y < m,
\\
 E'(y) = 1 - \cR_0 e^{m-y}, & y > m.
\end{array}
\]
Here $\cR_0$ is the basic reproduction number
without resistance as given by (\ref{eq:bas-rep-no}).

 $E$ is  strictly increasing on $\R_+$ if and only $\cR_0\le 1$. By Corollary \ref{re:jump-equiv}, $z$ makes a jump if and only if
$\cR_0 > 1$. By Theorem \ref{re:threshold-two-elab},
we also know that the jump occurs at $m$, and we have a lower estimate of its size.
For $m=0$, the case of the original Kermack/McKendrick model, the jump of the
final size curve is at zero
if $\cR_0 >1$ (as it is well-known).

Let $m > 0$ and $\cR_0 \le 1$. Then $E$ is strictly
increasing and convex on $\R_+$ with the convexity
being strict on $[m, \infty)$, $z(x) = E^{-1}(x)$
for $x \ge 0$. The inverse function $z$
is strictly increasing and concave.

Let $m > 0$ and $\cR_0 > 1$. Then $E$ is increasing
from 0 to $m$ as a straight line with $E(0) =0$
and $E(m)=m$. $E$ is strictly convex on
$[m, \infty)$ because $E'$ is strictly increasing
on that interval.
$E$ is strictly decreasing on $[m, m - \ln (1/\cR_0)]$ and strictly increasing on $\big [m - \ln (1/\cR_0), \infty \big )$ with $E(y) \to \infty$ as
$y \to \infty$. There is a unique $y_1 > m$
such that $E(y_1) = m$. $E$ is injective on
$[0,m]\cup (y_1, \infty)$. By (\ref{eq:z-right-inverse-min}),
\begin{equation}
\label{eq:final-size-curve-fixed}
z(x) = \left \{
\begin{array}{cc}
x , &  0 \le x \le m,
\\
E^{-1} (x) > y_1, & x >m .
\end{array}
\right .
\end{equation}
The final size curve $z$ jumps from $m$ to $y_1>m$ at $m$ and is concave on $(m,\infty)$ because
$E$ is convex on $(y_1,\infty)$.

The final cumulative force of infection, denoted as $z(x)$, is illustrated for  fixed resistance  $m=3$ and four values of $\mathcal{R}_0$ in Figure \ref{fig:Fixed-distrib}. We did not place this figure here but, for easier comparison, close to Figure \ref{fig:Gamma-distrib} which displays the
final cumulative force of infection for Gamma
distributed resistance with mean resistance 3.
Figure \ref{fig:Fixed-distrib} shows how the sizes of the  jumps,
 which occur at the resistance value $m=3$,
depend on the basic reproduction number $\cR_0$.


\subsection{More on jumps}
\label{subsec:abs-jumps}

As in Section \ref{exp:fixed}, we would not only
like to know whether a jump occurs but also
where is occurs and how large it is.

In the following, we explore various scenarios
which may overlap and not cover all possibilities.

Remember the function $E: \R_+ \to \R$  defined by
$E(x) = x -F(x)$. See Remark \ref{re:right-inverse} and
\begin{equation}
\label{eq:discon}
z(x) = \min E^{-1}(x), \quad x \in \R_+,
\end{equation}
with $E^{-1}(x) =\{y \in \R_+; E(y)=x \}$.

\begin{theorem}
\label{re:sol-nojumps-app}
Assume that $F$ is  differentiable on $\R_+$
and $F'(x) <1$ for all $x \in \R$.
Then the final size curve $z: \R_+ \to \R_+$
satisfying (\ref{eq:recursion-limit-equation-app}),
$z(0) =0$, is the inverse function of $E$, is  differentiable
and satisfies the differential equation,
\begin{equation}
\label{eq:final-size-curve-der}
z'(x) = (E'(z(x))^{-1}= (1- F'(z(x)))^{-1}, \qquad x \in \R_+.
\end{equation}
\end{theorem}

We mention that the assumption for $F$ is
satisfied if $F(y) = \cR_\di \tilde F(y)$ for all $y
\in \R$ (Remark \ref{eq:F-dist-R0})
and $\tilde F$ has a bounded derivative and
$\cR_\di < ( \sup \tilde F'(\R_+))^{-1}$.

\begin{proof} It follows that $E$ is differentiable and $E'(x) >0$ for all $x \in \R$.
So $E$ is strictly increasing and $z$ is the inverse function of $E$. Apply the one-dimensional
inverse function theorem.
\end{proof}

\begin{theorem}
\label{re:sol-jumps01}
Assume that $F$ is bounded on $\R_+$.
Let there are $0 \le z_1 < z_2 < \infty$
such that $E(z_1) > E(z_2)$.
Let
\begin{equation}
\label{eq:jump-point}
x_1 = \max E([0,z_2]).
\end{equation}
 Then $x_1 \ge 0$,
$z$ is discontinuous at $x_1$ and $z(x_1) < z_2$
and $z(x) \ge z_3$ for $x \in (x_1,\infty)$
where $z_3$ is the smallest $y> z_2$ such that
$E(y) =x_1$.
\end{theorem}

\begin{remark} Let $F$ satisfy Remark \ref{eq:F-dist-R0}. Then $E(y) = y - \cR_\di \tilde F(y)$. Let there are $ 0 \le z_1 < z_2 < \infty$
such that $\tilde F(z_1) > \tilde F(z_2)$.
Then the assumption for $F$ in Theorem \ref{re:sol-jumps01}
is satisfied if $\cR_\di $ is large enough,
\[
\cR_\di > \frac{z_2- z_1}{\tilde F(z_2) - \tilde F(z_1)}.
\]
\end{remark}

\begin{proof}
Since $E$ is continuous and $E(0) \ge 0$,
$x_1$ in (\ref{eq:jump-point}) exists and $x_1 \ge 0$.
Further, by definition,  $E(y) \le x_1$ for all $y \in [0,z_2]$. Since $E$ is continuous, there
is some $y\in [0,z_2]$ such that $E(y) =x_1$.
Since $x_1 \ge E(z_1) > E(z_2)$, $y\in [0,z_2)$. By (\ref{eq:discon}), $z(x_1) \le y <z_2$.
Further, $x_1 > E(z_2)$. Since $E$ is continuous
and $E(\tilde y) \to \infty$ as $\tilde y \to \infty$, there is some $z_3 > z_2$ such that $E(z_3)=x_1$ and $E(\tilde y) < x_1$ for $\tilde y
\in [z_2,z_3)$.

Let $x > x_1$. Then $E(y) <x$ for all $y \in [0,z_3]$. Since $E(z(x)) =x$ for all $x \in \R_+$,
$z(x) \ge z_3$.
\end{proof}

\begin{theorem}
\label{re:sol-jumps}
Assume that $F$ is bounded on $\R_+$.
Assume that $F$ is differentiable at some point
$z_2 \in (0,\infty)$ and $F'(z_2) > 1$.
 Then the final curve $z$
is discontinuous at some point $x_1 \in [0,z_2)$.
\end{theorem}

\begin{remark}
\label{re:point-discon}
Remember the function $E: \R_+ \to \R$  defined by
$E(x) = x -F(x)$. See Remark \ref{re:right-inverse}.
Then $x_1= \max E ([0,z_2])$
is a point of discontinuity for the
final size curve $z$. More precisely,
$z(x_1) < z_2$ and, for all $x > x_1$,
 $z(x) \ge z_3$ where $z_3$ is the smallest $y > z_2$ with $E(y)=x_1$. See
Theorem \ref{re:sol-jumps01}.
Further, $x_1$ is a local maximum of $E$
and  $z(x_1)$ is the smallest point in $[0,z_2]$ at which $E$ takes this local maximum.
\end{remark}

\begin{proof}
Recall the function $E: \R_+ \to \R$  defined by
$E(x) = x -F(x)$. See Remark \ref{re:right-inverse}.
Since $F'(z_2) >1$, $E'(z_2)<0$.
Then there exists some $\delta \in (0,z_2)$ such
that
\begin{equation}
\label{eq:no-extreme}
\begin{split}
& E(z) > E(z_2), \qquad z \in (z_2 - \delta, z_2),
\\
& E(z) < E(z_2), \qquad z \in (z_2, z_2 +\delta).
\end{split}
\end{equation}
See \cite[Cor.5-5]{Kir}, e.g.
In particular, $E$ is not increasing and $z$
makes a jump somewhere by Corollary \ref{re:jump-equiv}.
 The claim of  Theorem \ref{re:sol-jumps} now follows from
Theorem \ref{re:sol-jumps01} and the claim of Remark \ref{re:point-discon} follows from Theorem \ref{re:sol-jumps01}.
\end{proof}

\begin{remark}
\label{re:threshold-jumps-exist}
If $F$ is differentiable on $\R_+$,
\begin{equation}
\label{eq:threshold-jumps-exist}
\cT = \sup F'(\R_+)
\end{equation}
is a threshold between existence and nonexistence
of jumps of the final size curve: no jumps if $\cT<1$ (Theorem \ref{re:sol-nojumps-app})
and at least one jump if $\cT > 1$
(Theorem \ref{re:sol-jumps}).
\end{remark}

\begin{theorem}
\label{re:sol-jumps02}
Assume that $F$ is bounded on $\R_+$ and
$E$ is  increasing but not strictly increasing.
 Then there exist  $0 \le z_1 < z_2 < \infty$
such that $E$ is constant on $[z_1,z_2]$.
Let $x_1 = \max E([0,z_2])$. Then $x_1 \ge 0$
and
$z$ is discontinuous at $x_1$ and $z(x_1) \le z_1$
and $z(x) > z_2$ for $x \in (x_1,\infty)$.
\end{theorem}

\begin{remark}
\label{re:sol-jumps02-rare}
 If $F$ is of the form in
Remark \ref{re:F-R0}, the assumption for $F$
in Theorem \ref{re:sol-jumps02} is satisfied
if and only if $\tilde F(y) - \tilde F(z_1) = \cR_\di^{-1} (y-z_1)$ for all $y \in [z_1, z_2]$.
So an $\cR_\di$ for which this theorem applies
 only exists for exceptional $\tilde F$ and is
 uniquely determined.
 \end{remark}

\begin{proof}
Since $E$ is increasing, $x_1= E(z_1) = E(z_2)$.
By (\ref{eq:discon}), $z(x_1) \le z_1$.
Further, $E(y) \le x_1$ for all $y \in [0,z_2]$.
Let $x > x_1$. Since $E(z(x)) =x $ for all $x \in
\R_+$, $z(x) > z_2$.
\end{proof}

Finally, we explore how the final size curve behaves at 0.

\begin{theorem}
\label{re:behav-zero}
Assume that $F$ is twice continuously differentiable in a neighborhood
of 0 and
and $F'(0) <1$ and that $F'$ is differentiable at 0.
Then the  final size curve $z$ is continuously differentiable in a neighborhood 0 and $z'$ is differentiable at 0. Further
\[
\frac{z(x) - x}{x^2} \to (1/2) (1- F(0))^{-3} F^{\prime \prime}(0), \qquad x \to 0.
\]
\end{theorem}

\begin{proof}
The continuous differentiability of $z$ in a neighborhood of 0
follows as in the proof of Theorem \ref{re:sol-nojumps-app}
with (\ref{eq:final-size-curve-der}) holding in a neighborhood of 0.
By the chain rule, $z'$ is differentiable in a neighborhood of 0 and
\[
z^{\prime \prime}(x) = (1- F(x))^{-3} F^{\prime \prime}(z(x)).
\]
Apply l'H\^opital's rule twice,
\[
\frac{z(x) - x}{x^2} \to (1/2) z^{\prime \prime}(0), \quad x \to 0. \qedhere
\]
\end{proof}


\subsection{Intervals of convexity/concavity
for the final size curve}
\label{subsec:abs-convexity}

We are interested in the shape of the final size curve $z$, the minimal solution of (\ref{eq:final-size-function}).
For special cases, the previous results can be used
to obtain it from the graph of $E$
and (\ref{eq:discon}).
Alternatively,
we can use a convexity/concavity analysis.

\begin{theorem}
\label{re:convex-inherit}
Let $0 \le x_1 < x_2 \le \infty$. If $F$ is convex (concave)
 on the interval $[x_1, z(x_2))$, then the minimal solution $z$ of
(\ref{eq:recursion-limit-equation-app}) is convex (concave)
on the interval $[x_1,x_2)$ and continuous on $(x_1,x_2)$.
\end{theorem}

\begin{proof}
We show by induction that $z_n$ is convex (concave) on $[x_1,x_2)$ for all $n \in \N$,
 where the $z_n$ are given recursively by
(\ref{eq:recurs-functions-app}). Since convexity and concavity of functions
are preserved under pointwise convergence, $z$ is convex (concave)
on $[x_1,x_2)$ and thus continuous on $(x_1,x_2)$ \cite[Cor.6.3.3]{Dud}.

Notice that $z_0$ is both convex and concave.
Assume that $F$ is convex on $[x_1, z(x_2))$.

We provide the induction step.

Let $ n \in \N$ and $z_{n-1}$ be convex on $[x_1, z(x_2))$.
For $x \in [x_1,x_2)$, $z_{n-1}(x) \ge x \ge x_1$
and $z_{n-1}(x) \le z(x)< z(x_2)$. Let $\zeta\in [0,1]$
and $ y \in [x_1,x_2)$. By the standard definition of
convexity,
\[
z_{n-1}(\alpha x + (1-\alpha) y)
\le \alpha z_{n-1}(x) + (1-\alpha) z_{n-1}(y).
\]
Since $F$ is increasing and convex on $[x_1, z(x_2))$,
\[
\begin{split}
& F \big (z_{n-1}(\alpha x + (1-\alpha) y) \big)
\le
F \big (\alpha z_{n-1}(x) + (1-\alpha) z_{n-1}(y)\big )
\\
\le & \;
\alpha F(z_{n-1}(x))+ (1-\alpha) F(z_{n-1}(y)).
\end{split}
\]
By (\ref{eq:recurs-functions-app}), $z_n (\alpha x +(1-\alpha) y)
\le \alpha z_n(x) + (1-\alpha) z_n(y)$.

This shows that $z_n$ is convex on $[x_1,z(x_2))$.

The concave case is shown similarly.
\end{proof}

Proposition \ref{re:minimal-sol-app} and Theorem \ref{re:convex-inherit} also hold if $F$ is given by
 (\ref{eq:measure}). This does no longer hold for the subsequent results.

\begin{theorem}
\label{re:sol-jumps-convex-app}
Assume that $F$ is  differentiable on $\R_+$.
Let $z_3>z_2  > 0$ such that $F$ is convex on $[0, z_3)$
and $F'(z_2) \ge 1> F'(0)$.

Then there exists some $z_1 \in (0,z_2]$ such that $F'(z_1) =1$
and some $x_1 >0$ such that the solution $z$ of (\ref{eq:recursion-limit-equation-app})
 is differentiable on $[0,x_1)$, $z'(0)=(1- F'(0))^{-1}$, and convex and continuous
on $[0,x_1]$ with $z'(x) \to \infty$ as $x \to x_1-$
and $z(x) < z_1$ for $x < x_1$, $z(x_1)=z_1$ and $z(x) \ge z_3$ for $x > x_1$.

In particular, $z$ makes a jump at $x_1$ of the size $z_3 -z_1$
or larger.

Finally, on $[0,x_1)$, $z$ satisfies the initial value problem $z(0) =0$,
\[
z'(x) = (1- F'(z(x)))^{-1}, \qquad x \in [0,x_1).
\]
\end{theorem}

\begin{proof}
Since $F$ is convex and differentiable on $[0,z_2]$,
$F'$ is increasing on $[0,z_2]$ \cite[Cor.6.3.3]{Dud}. By Darboux's intermediate value theorem for derivatives \cite[6.2.9]{Lay}, there is some $z_1 \in (0,z_2)$ such that
$F'(z_1) =1$. Since $F'$ is increasing on $[0,z_2]$, $z_1$
can be chosen
such that
\begin{equation}
\label{eq:blow-up-prep}
F'(z) < 1, \quad z \in [0,z_1), \qquad F'(z_1)=1.
\end{equation}
Then $E$ is strictly increasing on $[0,z_1]$
and the restriction of $E$ to $[0,z_1]$ has
an inverse function.

Choose some $x_1 >0$ such that $z(x) < z_1$ for $x \in [0,x_1)$
and $z(x) >z_1$ for $x> x_1$. This is possible because $z$
is strictly increasing by Theorem \ref{re:minimal-sol-app}.
Since $z(x_1) \le z_1 < z_2$
 because $z$ is left continuous and since $F$ is convex
on $[0,z_2]$, $z$ is convex on $[0,x_1]$ by Theorem
\ref{re:convex-inherit}. Actually, $x_1 = E(z_1)$.
Since $E(z(x))=x$ for all $x \in \R$, the restriction of $z$ to $[0,x_1]$ is the inverse
function of the restriction of $E$ to $[0,z_1]$.
Since $E'(z) < 0$ for all $z \in [0,z_1)$,
$z$ is continuous on $[0,x_1]$ and differentiable on $[0,x_1)$ and
$z'(x) = (E'(z(x)))^{-1}$ for all $x \in [0,x_1)$
by the one-dimensional inverse function theorem.
Since $E'(z_1)=0$, $z'(x) \to \infty$ as $x \to x_1-$.

Suppose that $z(\tilde x) < z_3$ for some $\tilde x > x_1$.
Then $z(x) < z_3 $ for $x \in [0, \tilde x]$.
By Theorem \ref{re:convex-inherit}, $z$ is convex on $[0, \tilde x).$ If one draws the graph of $z$, this is impossible.
For a rigorous argument,  the Dini derivatives of $z$ exist
and are increasing
on $(0,\tilde x)$ and equal $z'$ on $(0,x_1)$
\cite[(17.2), (17.37)]{HeSt} \cite[Cor.6.3.3]{Dud}.
But this contradicts that $z'(x) \to \infty$ as $x \to x_1-$.
\end{proof}

\begin{remark}
If $F'(z_2) > 1$, by  Remark \ref{re:point-discon}, $ \sup E ([0,z_2])$ is also a point of discontinuity for $z$. In many cases, $x_1$ and $\sup E ([0,z_2])$ will
coincide, but it may be possible to construct
examples where they are not. If they coincide,
we may get two different estimates for the
size of the jump.
We guess that  Remark \ref{re:point-discon}
provides a larger size than Theorem \ref{re:sol-jumps-convex-app} which, as a trade-off, provides
extra information about the shape of the
final size curve on $[0,x_1)$.
\end{remark}

Recall that the derivative of a convex
differentiable function is continuous
because it is increasing \cite[Cor.6.3.3]{Dud} and has the intermediate
value property \cite[6.2.9]{Lay}.

\begin{theorem}
\label{re:sol-jumps-convex-app1}
Assume that $F$ is  differentiable on $\R_+$.
Let $z_2  > 0$ such that $F$ is convex on $[0, z_2]$
and $F'(z_2) > 1> F'(0)$.

Then there exists some $z_1 \in (0,z_2)$ such that $F'(z_1) =1$
and some $x_1 >0$, $x_1 =E(z_1)$, such that the solution $z$ of (\ref{eq:recursion-limit-equation-app})
 is differentiable on $[0,x_1)$, $z'(0)=(1- F'(0))^{-1}$, and convex and continuous
on $[0,x_1]$ with $z'(x) \to \infty$ as $x \to x_1-$
and $z(x) < z_1$ for $x < x_1$, $z(x_1)=z_1$ and $z(x) > z_2$ for $x > x_1$.

In particular, $z$ makes a jump at $x_1$ of the size $z_2 -z_1$
or larger.

Finally, on $[0,x_1)$, $z$ satisfies the initial value problem $z(0) =0$,
\[
z'(x) = (1- F'(z(x)))^{-1}, \qquad x \in [0,x_1).
\]
\end{theorem}



\section{More on the final size  for distributed
resistance}
\label{sec:dist-final-size}



Guided by Section \ref{sec:abst-fin-size} and \ref{sec:distributed}, we apply the results
of the previous section to the final size
equation (\ref{eq:final-size-function}) with the map
\begin{equation}
\label{eq:map-dis}
F(y)  = \cR_\di \int_0^{y}  \sigma( m) f (y - m) dm, \qquad y \in\R_+.
\end{equation}
Here $\cR_\di$ is the basic reproduction number
without resistance as given by (\ref{eq:Xi}).
Assume that $\sigma:\R_+\to \R_+$ is continuous on $\R_+$,
\begin{equation}
\label{eq:sigma-integral}
\int_0^\infty \sigma(m) dm =1.
\end{equation}
By Leibniz rule, since $f(0)=0$,
\begin{equation}
\label{eq:Leibniz}
F'(y) = \cR_\di \int_0^y \sigma(m) f'(y-m) dm
=
\cR_\di \int_0^y \sigma(m) e^{m-y} dm \ge 0.
\end{equation}
So $F$ is  increasing, $F(0)=0= F'(0)$,
\begin{equation}
\label{eq:map-asympt}
\begin{split}F(y) \le & \cR_\di f(y) \le \cR_\di, \qquad y \in \R_+,
\\
F(0) = & 0, \qquad F(y) \to \cR_\di, \qquad y \to \infty,
\end{split}
\end{equation}
again by Lebesgue's dominated convergence theorem.

For intermediate $y>0$, $F'(y) >0$,
\begin{equation}
\label{eq:map-deriv-asympt}
 F'(0)=0, \qquad F'(y) \to 0, \quad y\to \infty,
\end{equation}
 by (\ref{eq:Leibniz}) and
 Lebesque's
theorem of dominated almost everywhere convergence.
Recall (\ref{eq:sigma-integral}).
So $F'$ is neither increasing nor decreasing, and $F$ is neither convex nor concave on all of $\R_+$.
Further, by (\ref{eq:Leibniz}) and
(\ref{eq:sigma-integral}),
$F'$ is bounded and
\begin{equation}
\label{eq:F'bounded}
F'(y) < \cR_\di, \qquad y \in \R_+.
\end{equation}

Recall that a differentiable function on an
open interval is
convex if and only if its derivative is increasing
\cite[p.43]{Rud}.

\begin{lemma}
\label{re:convex-cave-F}
(a) If $m_2 > 0$ and $\sigma( m) =0$ for $m \ge m_2$, then
$F'$ is decreasing on $[m_2, \infty)$ and $F $ is concave
on $[m_2, \infty)$.

(b) If $m_1 >0$ and $\sigma$ increases on $[0, m_1]$,
then $F'$ increases and $F$ is convex on $[0,m_1]$.
\end{lemma}

As for (b), after a substitution,
\begin{equation}
\label{eq:Leibniz-subst}
F'(y) = \cR_\di \int_0^y \sigma(y-m) f'(m) dm
=
\cR_\di \int_0^y \sigma(y-m) e^{-m} dm
\end{equation}
If $\sigma$ is increasing on $[0, m_1]$,
so if $F'$ and $F$ is convex on $[0,m_1]$.

Using Leibniz rule again and (\ref{eq:Leibniz}), $F'$ is differentiable and
\begin{equation}
\label{eq:Leibniz2}
F^{\prime \prime}(y) = \cR_\di \Big (\sigma(y) - \int_0^y \sigma(m) e^{m-y} dm \Big ).
\end{equation}

By Theorem \ref{re:behav-zero}, the final size solution $z$
of $z(x) = F(z(x)) +x$, $x \in \R_+$,
behaves as follows close to 0.
Recall that $x$ is the final size of the cumulative force of infection that is due to the initial infectives while $z(x)$ is the final size of the
cumulative force of  infection exerted by all infectives.

\begin{theorem}
\label{re:behav-zero-dis}
\[
\frac{z(x) - x}{x^2} \to \cR_\di \sigma(0)/2, \qquad x \to 0.
\]
\end{theorem}

We obtain the following results for the convexity/concavity
behavior of the final size curve from Theorem
\ref{re:convex-inherit} and Lemma \ref{re:convex-cave-F}.

\begin{theorem}
\label{re:fin-size-convex-cave}
Let $z$ be the final size curve $z(x) =
F(z(x)) +x$, $x \in \R_+$, and $m_1 \in (0, \infty)$.

\begin{itemize}
\item[(a)] $z$ is convex on $[0,m_1]$
if $\sigma $ increases on $[0, z(m_1)]$
and in particular if $\sigma$ increases on $[0, m_1 +\cR_\di]$.

\item[(b)] If $m_2 > 0$ and $\sigma(m) =0$ for $m \ge m_2$,
then $z$ is concave on $[m_2, \infty)$.
\end{itemize}
\end{theorem}

Notice that $z(x) \le x + \cR_\di$, $x \in \R_+$,
by (\ref{eq:map-asympt}).

The following result is obtained
from Theorem \ref{re:sol-nojumps-app}.

\begin{theorem}
\label{re:derivative}
Assume that
\[
\cR_\di \int_0^y \sigma (m) e^{m-y} dm <1, \qquad y \in \R_+.
\]
Then the solutions $J(\infty)=z$ of (\ref{eq:model-Volt-int})
are uniquely determined
by  $J_0(\infty)=x \in \R_+$ and are a twice continuously differentiable
function of $J_0(\infty)$,
\begin{equation}
\label{eq:derivative-final}
z'(x) = \Big ( 1 - \cR_\di \int_0^{z(x)} \sigma(m) e^{m-z(x)} dm \Big )^{-1}.
\end{equation}
In particular, $z'(0) =1$ and $z'(x) \to 1$ as $x \to \infty$,
$z^{\prime \prime}(0)= \cR_\di \sigma(0)$.
\end{theorem}

By Theorem \ref{re:sol-jumps}, we have the following result.

\begin{theorem}
\label{re:distr-jump-exist}
Assume that there is some $z_2 \in \R_+$ such that
\begin{equation}
\cR_\di \int_0^{z_2} \sigma(m) e^{m-z_2} dm >1.
\end{equation}
Then the final size curve $z$ has a discontinuity.
\end{theorem}

\subsection{A convexity interval for the
final size curve}

\begin{theorem}
\label{re:density-creates-jump}
Assume that there is some $z_2 \in \R_+$ such that
$\sigma$ increases on $[0,z_2]$ and
\begin{equation}
\cR_\di \int_0^{z_2} \sigma(m) e^{m-z_2} dm >1.
\end{equation}
Then  there is some unique $z_1 \in (0,z_2)$ such that
\begin{equation}
\begin{split}
& \cR_\di \int_0^{z_1} \sigma(m) e^{m-z_1} dm =1.
\end{split}
\end{equation}

There also exists some unique $x_1 \in [0,z_1]$
with the following property:

$z$ is twice differentiable on $[0,x_1)$, convex and continuous
on $[0,x_1]$, $z(x_1) = z_1$.
Further,
$z'(0) =1$ and
$z^{\prime \prime}(0)= \cR_\di \sigma(0)$.
Most importantly,
$z'(x) \to \infty$ as $x \to x_1-$
and  $z(x) > z_2$ for $x > x_1$.

In particular, $z$ makes a jump at $x_1$ of the size $z_2 -z_1$
or larger.
\end{theorem}

\begin{proof}
We apply Theorem \ref{re:sol-jumps-convex-app1}.

By Lemma \ref{re:convex-cave-F}, $F$ is convex on
$[0,z_2]$. By (\ref{eq:Leibniz-subst}),
\begin{equation}
\label{eq:Leibniz-subst-rep}
F'(y) =
\cR_\di \int_0^y \sigma(y-m) e^{-m} dm >0, \qquad y \in \R_+,
\end{equation}
and so $0= F'(0)<1 < F'(z_2)$. Existence of $z_1 \in (0,z_2)$
with $F'(z_1) =1$ follows from the intermediate value theorem
and the continuity of $F'$ (or from the intermediate value
theorem for derivatives). Since $\sigma$ is increasing on $[0,z_1]$, $\sigma(z_1) >0$ and the continuous $\sigma$ is strictly positive
in a neighborhood of $z_1$. By (\ref{eq:Leibniz-subst-rep}),
$F'$ is strictly increasing in a neighborhood of $z_1$
and $z_1$ is uniquely determined by $F'(z_1) =1$.
Since $z$ is strictly increasing and continuous from the left,
there exists a unique $x_1 \in (0,z_1)$ such that $F(x_1) =z_1$,
$F(x) < z_1$ for $x \in (0,x_1)$ and $F(x) > z_1$ for $x > x_1$.

The remaining statements follow from Theorem
\ref{re:sol-jumps-convex-app1}.
\end{proof}

\subsection{A new threshold parameter}
\label{subsec:new-thres}

By Theorem \ref{re:derivative} and \ref{re:distr-jump-exist},
\begin{equation}
\label{eq:threshold-jumps-exist-dist}
\cT = \cR_\di \max_{y \ge 0} \int_0^{y} \sigma(m) e^{m-y} dm
< \cR_\di
\end{equation}
is a threshold parameter between existence and
nonexistence of jumps of the final size curve:
No jumps if $\cT <1$ and at least one jump if
$\cT > 1$. Cf. Remark \ref{re:threshold-jumps-exist}.
$\cT$ can be considerably smaller
than $\cR_\di$. In the next section, we
present a Gamma-distributed resistance for which
$\cT =1$ and  $\cR_\di =4.46$.
Recall that $\cR_\di$ is the basic reproduction
number without resistance. It seems difficult
to interpret $\cT$ as a basic reproduction
number with resistance.



\section{A very special example with Gamma
distributed resistance}
\label{sec:example-Gamma}


To illustrate the results of the previous section,
we assume that the resistance distribution
of the initial susceptibles is a  special Gamma distribution,
\begin{equation}
\label{eq:Gamma}
\sigma(m) = \beta m^{\alpha} e^{-m}, \qquad m \ge 0,
\end{equation}
with some $\alpha > -1$ and
$\beta$ be chosen such that $\int_0^\infty \sigma(m) =1$,
\begin{equation}
\label{eq:Gamma-scale}
\beta=\frac{1}{\Gamma(\alpha+1)}
\end{equation}
 with $\Gamma$
being the Gamma function.
Notice that the mean resistance is $\alpha +1$.
We differentiate,
\[
\sigma'(m) = \beta (\alpha m^{\alpha -1} - m^\alpha)e^{-m}, \qquad m \ge 0.
\]
So, $\sigma$ is strictly increasing on $[0, \alpha]$ and
strictly decreasing on $[\alpha, \infty)$.
By (\ref{eq:Leibniz}),
\begin{equation}
\label{eq:F'-Gamma}
\begin{split}
F^\prime(y) = &\beta \cR_\di \int_0^y \sigma(m) f'(y-m) dm=
\beta \cR_\di \int_0^y m^\alpha  e^{-y} dm
\\
= &
\beta \cR_\di (\alpha +1)^{-1} y^{\alpha+1}  e^{-y},
\qquad y \in \R_+.
\end{split}
\end{equation}
Here $\cR_\di$ is the basic reproduction number without resistance given by (\ref{eq:Xi}) and $\beta$ is given by
(\ref{eq:Gamma-scale}).
Notice that $F'$ is proportional to a Gamma distribution.
We differentiate once more,
\[
F^{\prime \prime}(y) = \beta \cR_\di \big (y^\alpha - (\alpha +1)^{-1} y^{\alpha+1} \big )e^{-y},\qquad y \ge 0.
\]
So, $F'$ takes its maximum at $\alpha +1$
and $F$ is strictly convex on $[0, \alpha +1]$ and strictly concave on
$[\alpha +1, \infty)$.
 The threshold parameter
for existence of jumps, (\ref{eq:threshold-jumps-exist}), is
\begin{equation}
\label{eq:threshold-par-Gamma}
\cT = \max F'(\R_+)
    = F'(\alpha+1) = \cR_\di \beta  (\alpha+1)^{\alpha}  e^{-(\alpha+1)}.
\end{equation}
For this example, we  have obtained sharper results
concerning the domains of convexity and concavity
of $F$ than they are provided by Section \ref{sec:dist-final-size}.

Consequently, we will get sharper results
for the final size curve from Section
\ref{sec:abst-fin-size} than from
Section \ref{sec:dist-final-size}.

Let $E: \R_+ \to \R$ be defined as in
(\ref{eq:functionE})  by
\begin{equation}
\label{eq:E-illus}
E(y) =y - F(y), \qquad y \in \R_+.
\end{equation}
Notice that
$E(y) < y$ for $y >0$. Since $F(0)=0$ and $F$
is bounded, $E(\R_+) \supseteq \R_+$.
 $E$ is infinitely often differentiable,
and its graph has an inflection point at $\alpha +1$.
More precisely, $E$ is strictly concave on $[0, \alpha+1)$ and strictly convex on $(\alpha+1, \infty)$.

\subsection{Subthreshold case}
\label{subsec:subth}
Assume that
\[
\cT = \max F' = F'(\alpha+1) = \cR_\di \beta  (\alpha+1)^{\alpha}  e^{-(\alpha+1)}
<1.
\]
 By Theorem \ref{re:sol-nojumps-app}, the final size curve $z$, $z(x) = x + F(z(x))$, $x \in \R_+$, is twice continuously differentiable. By (\ref{eq:map-asympt}),
 $z(x) \le x + \cR_\di$, $x \in \R_+$.

Recall that $x$ is the final size of the cumulative force of infection that is due to the initial infectives while $z(x)$ is the final size of the
cumulative force of  infection exerted by all infectives.

Most importantly, the function $E$ is strictly
increasing and $z$ is the inverse function of $E$.
So, $z$ is strictly convex on $\big[0, E(\alpha+1)\big )$
and strictly concave on $\big(E(\alpha+1), \infty \big)$.

Finally, $z$ is a solution of the initial value
problem $z(0) =0$,
\[
z'(x) = (1 - F'(z(x))^{-1}, \quad x \in \R_+.
\]


\subsection{Superthreshold case}
\label{subsec:dist-superth}

Assume that
\begin{equation}
\label{eq:Gamma-superthr}
\cT = \cR_\di \beta  (\alpha+1)^{\alpha}  e^{-(\alpha+1)}
>1.
\end{equation}
Let $z: \R_+ \to \R_+$ be the final size curve,
$z (x) = F(z(x)) +x$, $x \in \R_+$.

By Theorem \ref{re:sol-jumps-convex-app}, there exists a unique $z_1 \in (0, \alpha+1)$
such that
\[
\beta \cR_\di (\alpha +1)^{-1} z_1^{\alpha+1}  e^{-z_1} =1
\]
and a unique $x_1 \in (0,z_1)$ such that
$z$ is differentiable on $[0,x_1)$, convex and continuous
on $[0,x_1]$, $z(x_1) = z_1$.
Further,
$z'(x) \to \infty$ as $x \to x_1-$
and  $z(x) > \alpha+1$ for $x > x_1$.

By Theorem \ref{re:convex-inherit}, $z$ is concave on $[\alpha +1, \infty)$.

Finally, on $[0,x_1)$, $z$ is a solution of the initial value
problem $z(0) =0$,
\[
z'(x) = (1 - F'(z(x))^{-1}, \quad x \in [0,x_1).
\]
In the superthreshold case, $E'(\alpha+1) < 0$.
Thus, $E$ takes a unique local maximum $x_1= E(z_1)$
at some point $z_1 \in (0,\alpha +1)$, $x_1 < z_1$, and a
unique local minimum at some point $z_2 > \alpha+1$.
In summary,
\begin{equation}
0 < x_1 < z_1 < \alpha +1 < z_2.
\end{equation}

 $E$ is strictly increasing on
 $(0,z_1)$, strictly decreasing on $(z_1,z_2)$,
 and strictly increasing on $(z_2, \infty)$.
There exists a unique $z_3 \in (z_2,\infty)$   such $x_1 = E(z_3)$. The restriction of $E$ to $[0, z_1]\cup (z_3,\infty)$ is an injective mapping
from this union onto $\R_+$.
By Remark \ref{re:right-inverse},
\begin{equation}
\label{eq:graph-finalsize-Gamma}
z(x) = \left \{
\begin{array}{rl}
 E^{-1}(x) \in [0,z_1], &\quad x \in [0,x_1],
 \\
 E^{-1}(x)> z_3, & \quad x > x_1.
\end{array}
\right.
\end{equation}
In particular, $z$ jumps at $x_1$ from $z_1$ to $z_3$ and is strictly convex on $[0,x_1]$
and strictly concave on $[x_1,\infty)$.
The last holds because $E$ is strictly concave on $(z_3, \infty)$.

Notice that the jump occurs at $x_1$ which is
strictly smaller than the mean resistance $\alpha +1$.


\subsection{Numerical illustrations}
\label{subsec:num-illus}


For the purpose of  numerical illustration,
by integrating $F'$  in (\ref{eq:F'-Gamma}) and using $F(0)=0$, with $\alpha =2$,  we obtain
\begin{equation}
\label{eq:F-example-Gamma}
F(y) = \beta \cR_\di \Big ( 2 - e^{-y}\big  ((1/3) y^3 + y^2 + 2y + 2 \big )\Big), \qquad y \in \R_+.
\end{equation}
Here, $\cR_\di$ is the basic reproduction number without resistance as given by
(\ref{eq:Xi}) and $\beta$ is the scaling factor (\ref{eq:Gamma-scale}).
For $\alpha =2$, the superthreshold case (\ref{eq:Gamma-superthr})
holds if and only if  $\cR_\di > 4.46$.

To plot the graph of the final size curve $z$,
we use the graph of $E$ with $F$ given by
(\ref{eq:F-example-Gamma}) and the inversion formula (\ref{eq:graph-finalsize-Gamma}).

\begin{figure}[ht]
\centering
\includegraphics[width=0.45\textwidth]{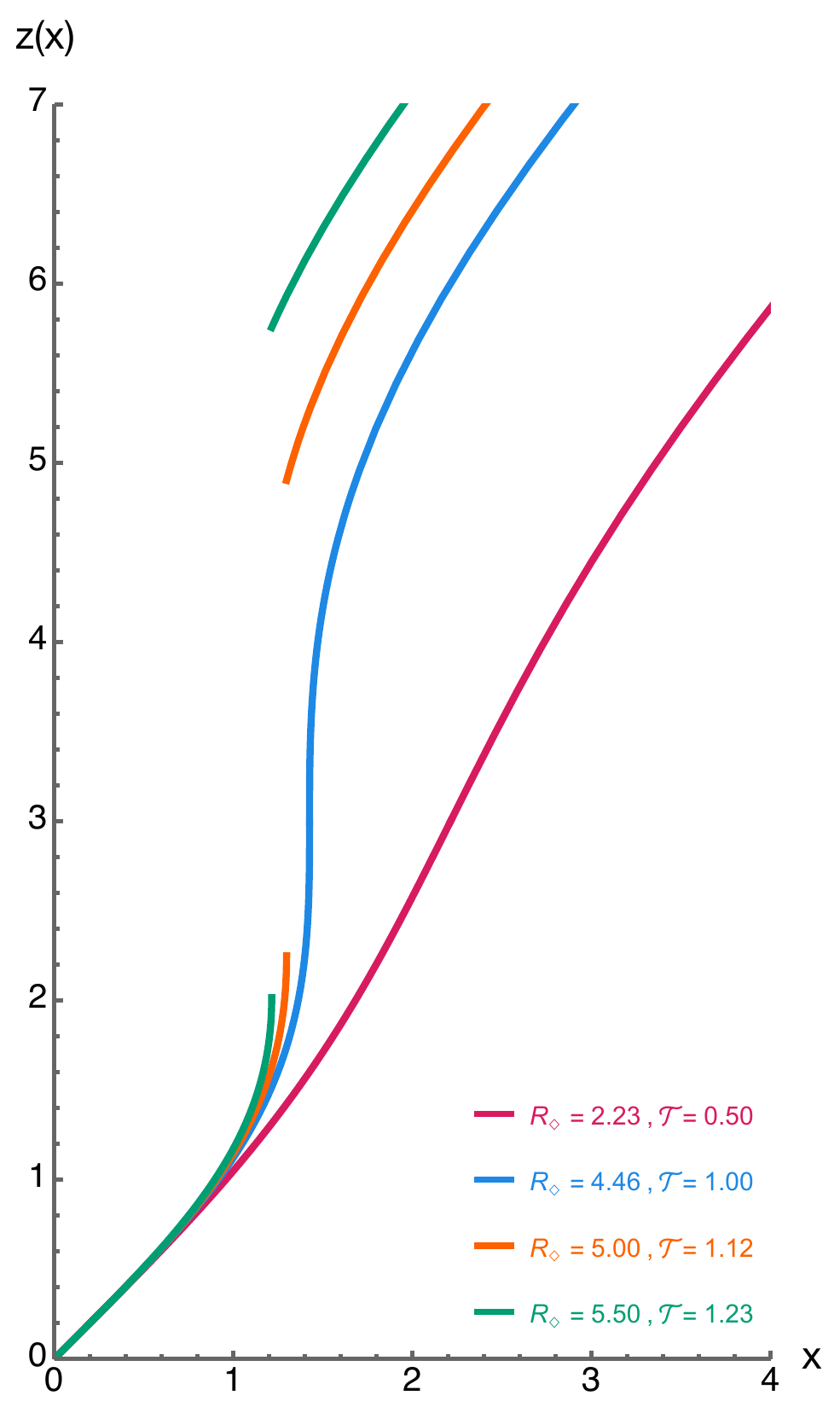}
    \caption{Plot of  the final size curve for Gamma-distributed
resistance with $\alpha =2$ and different
values of $\cR_\di$.}
    \label{fig:Gamma-distrib}
\end{figure}
 In Figure \ref{fig:Gamma-distrib}, we present  the final size
curve $z$ for four different values of $\cR_\di$.
Recall that the final size curve gives  the cumulative force of infection exerted by all infectives, $z(x)$, as a function of the cumulative force of infection exerted by the initial infectives, $x$.

Further, recall that $z(x) = \ln (S_\infty/ S_0)$
where $S_0$ is the initial size and $S_\infty$
the final size of the susceptible part of the
host population, (\ref{eq:sus-cumforce-fin-dist}).

\begin{figure}[ht]
    \centering
    \includegraphics[width=0.7\textwidth]
    {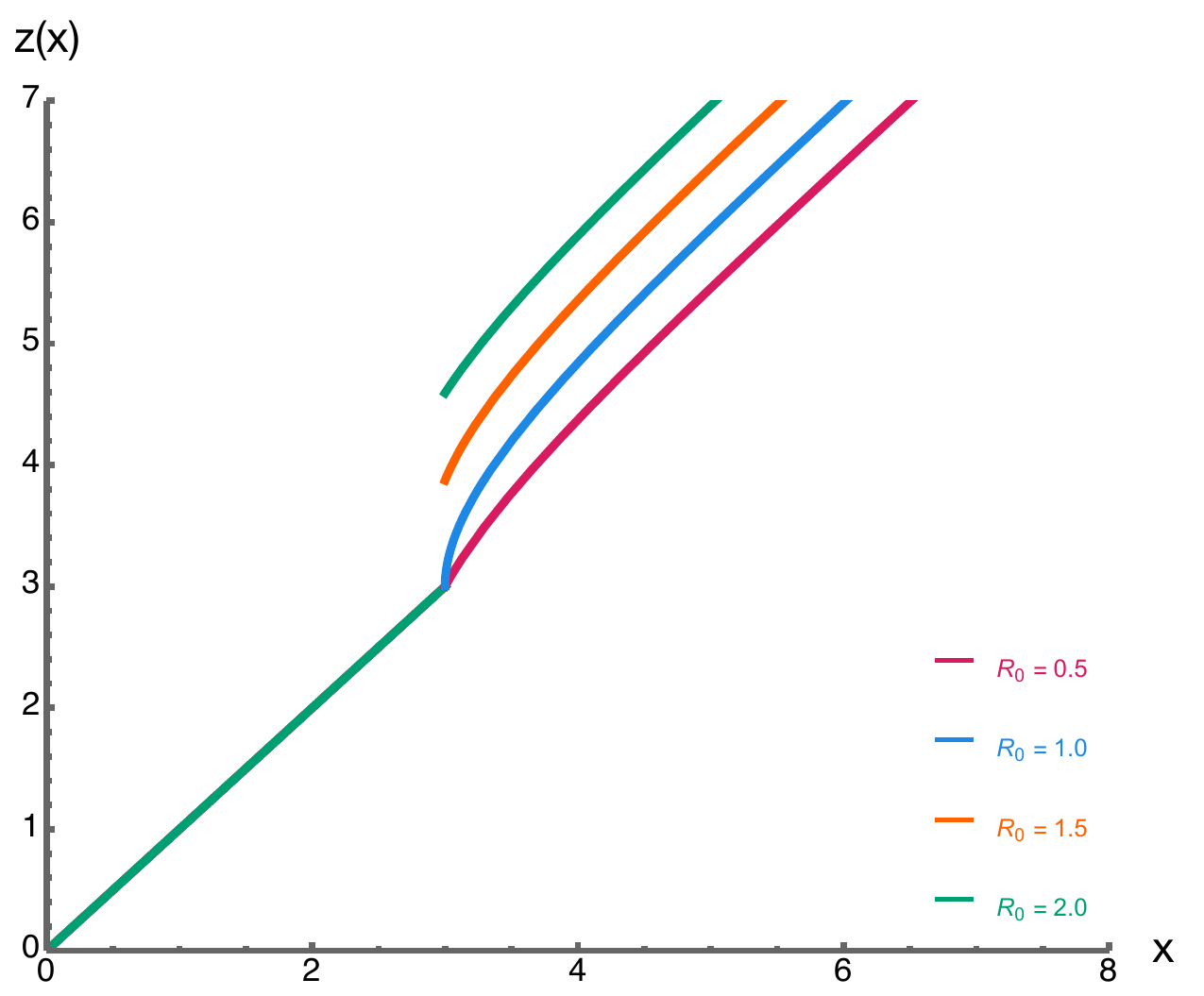}
    \caption{Plot of  the final size curve for fixed resistance $m=3$ and different values of $\mathcal{R}_0$.}
    \label{fig:Fixed-distrib}
\end{figure}

Figure \ref{fig:Gamma-distrib} should be
compared to Figure \ref{fig:Fixed-distrib}
where the resistance is fixed and has the value
of the mean resistance for Figure \ref{fig:Gamma-distrib} which is 3.
The graph of $z$ is obtained  in a similar way
as for Figure \ref{fig:Gamma-distrib}, but with
$F$ given by (\ref{eq:F-fixed}).
See Section \ref{exp:fixed}.

Notice that $\cR_0$
in Figure \ref{fig:Fixed-distrib}
and $\cR_\di$ on Figure \ref{fig:Gamma-distrib} are given
by the same formula,  (\ref{eq:bas-rep-no}) and  (\ref{eq:Xi}),
and have the interpretation as basic reproduction
number (without resistance). Different symbols
have been chosen because $\cR_0$ is the
relevant threshold parameter for fixed resistance
while $\cT$ rather than $\cR_\di$ is  the relevant
threshold parameter for distributed resistance.
 $\cT$  is given  by (\ref{eq:threshold-jumps-exist-dist}) in general
and by (\ref{eq:threshold-par-Gamma}) in this special case.

These figures nicely illustrate the importance and the effects of host  heterogeneity.
In our examples, we find that the threshold for the existence of  jumps of the
final size curve is (often substantially) higher for distributed resistance than for fixed resistance.

Further, for fixed resistance, the jump of the final
size curve always occurs when the
 final size of the initial  cumulative force of infection passes the (fixed) resistance level;
 in our example for  Gamma-distributed
 resistance, the jumps occur at smaller final sizes
 of the initial cumulative force of infection
 which are the smaller the larger the
 values of $\cT$.



\section{Discussion}
\label{sec:discuss}

In \cite{HoWa2, Wal},  a more general
 condition than the one involving
(\ref{eq:thres-cum-force})
is considered, where
a host that becomes infectious at time $t$
was  infected at a time $r$ with
\[
\int_r^t (\delta + I(s)) ds = m.
\]
If $\delta > 0$, every  infected individual
will become infectious  after a delay
$m /\delta$ at the latest, and our explanation
why there are different disease outcomes in
similar cities does
no longer work. A possible proportionality constant in front of $I$ \cite{Coo, HoWa1, HoWa2, Wal}
has been absorbed into the
resistance  $m$.

These threshold delay conditions are early examples of
state-dependent delays
and differential equations containing them
have become topics of
a rich theory (\cite{BGR, HaKWW, KrWa, MaNu, SmKu}
and the references therein). The specialty of
our threshold delay condition ($\delta =0$) allows us to
bypass the intricacies of this theory and
transform the problem into a single  integral
equation of Volterra type from which much desirable
information can be obtained.
One obtains a deterministic explanation of why there are different
epidemic outcomes in  cities of similar size, namely
that the final sizes of the primary infective force
or the disease resistance distributions of the initial
susceptibles were different.

In the version of the model with distributed
resistance, we have made the simplifying assumption
that the disease-mortality $P$ in (\ref{eq:sec-inf-age-dens})
and the
infectivity-age dependent infectious influence
of an infectious host  $\kappa$
in (\ref{eq:FoI-dist})
do not depend on its
resistance $m$. Of course, one could consider
a dependence instead. Such a generalization
would  bring greater
mathematical difficulty requiring functional
analysis instead of what mostly is advanced calculus.  See \cite{BoDI1} and \cite{BoDI2}
for examples of models that
incorporate  other distributed traits
 than distributed resistance.
Since we expect the essence of the epidemiological
message to be similar, we  use Occam's razor (principle of parsimony) for this  presentation.
In this spirit, we also stick to the equation
\begin{equation*}
S'(t) = - S(t)I(t)
\end{equation*}
rather than the more general equation
\begin{equation}
S'(t) = -\eta(S(t)) I(t)
\end{equation}
with $\eta:\R_+ \to \R_+$ continuous and
$\int_0^1 \frac{1}{\eta(s)}ds = \infty$
considered in \cite{Thi77}.

The epidemic model with host resistance to the disease
provides a handy theoretical laboratory for
studying importance and the effects of host  heterogeneity.
We consider both models with fixed resistance
and with distributed resistance.
We find that the threshold for the existence of  jumps of the
final size curve is (often substantially) higher for distributed resistance than for fixed resistance.

Further, for fixed resistance, the jump of the final
size curve always occurs when the
 final size of the initial  cumulative force of infection passes the (fixed) resistance level;
 in our example for  Gamma-distributed
 resistance, the jumps occur at smaller final sizes
 of the initial cumulative force of infection.

It is difficult to relate our threshold
parameters $\cR_0$ and $\cT$
to estimates of reproduction numbers found in the literature like in \cite{SwissScienceTaskForce}
because these estimates have been performed under
the assumption that there is no resistance
to the disease. In particular for distributed
resistance, it is not clear whether they should
be related to $\cR_\di$ or $\cT$.

 As expected, incorporating resistance towards
 the disease into  epidemic models creates an
 Allee effect for the final size of the force of
 infection. Less expected, the Allee effect is quite pronounced
 and even occurs in the form of jumps, whether the resistance is the
 same for all hosts or  is distributed over the
 host population.



\bibliography{}


\end{document}